\documentclass[reqno]{amsart}
\usepackage{mathrsfs}
\usepackage{graphicx}
\usepackage{amsmath}
\usepackage{amscd}
\usepackage{cite}
\usepackage{hyperref}
\usepackage{epsfig}
\usepackage{subfigure}
\usepackage{caption}
\usepackage{tikz}
  \usepackage{float}
  \usepackage{amssymb}
\usepackage{amsfonts}
\usepackage{appendix}
\usepackage{bm}
\usepackage{geometry}

\newtheorem{Theorem}{Theorem}[section]

\newtheorem{Definition}[Theorem]{Definition}
\newtheorem{Proposition}[Theorem]{Proposition}
\newtheorem{Lemma}[Theorem]{Lemma}

\newtheorem{Conjecture}[Theorem]{Conjecture}
\newtheorem{Remark}[Theorem]{Remark}

\newcommand{\bC}{{\mathbb C}}

\newcommand{\bP}{{\mathbb P}}

\newcommand{\bZ}{{\mathbb Z}}

\newcommand{\cB}{{\mathcal B}}

\newcommand{\cE}{{\mathcal E}}
\newcommand{\cF}{{\mathcal F}}

\newcommand{\cM}{{\mathcal M}}

\newcommand{\cP}{{\mathcal P}}

\newcommand{\half}{\frac{1}{2}}

\newcommand{\cW}{{\mathcal W}}

\newcommand{\Mbar}{\overline{\cM}}

\newcommand{\tc}{{\tilde c}}

\newcommand{\tF}{{\widetilde F}}

\newcommand{\tA}{{\widetilde{A}}}

\newcommand{\tGa}{{\widetilde{\Gamma}}}

\newcommand{\be}{\begin{equation}}
\newcommand{\ee}{\end{equation}}
\newcommand{\bea}{\begin{eqnarray}}
\newcommand{\ben}{\begin{eqnarray*}}
\newcommand{\een}{\end{eqnarray*}}
\newcommand{\eea}{\end{eqnarray}}

\DeclareMathOperator{\Aut}{Aut}
\DeclareMathOperator{\Id}{id}

\usepackage{color}

\definecolor{yellow}{rgb}{1,1,0}
\definecolor{orange}{rgb}{1,.7,0}
\definecolor{red}{rgb}{1,0,0}
\definecolor{green}{rgb}{0,1,1}
\definecolor{white}{rgb}{1,1,1}

\definecolor{A}{rgb}{.75,1,.75}

\begin{document}

\newtheorem{myDef}{Definition}
\newtheorem{thm}{Theorem}
\newtheorem{eqn}{equation}

\title[On a Proof of the ADKMV Conjecture]
{On a Proof of the ADKMV Conjecture\\
---
$3$-KP Integrability of the Topological Vertex}

\author{Zhiyuan Wang}
\address{Zhiyuan Wang, School of Mathematics and Statistics,
	Huazhong University of Science and Technology,
	Wuhan, China}
\email{wangzy23@hust.edu.cn}

\author{Chenglang Yang}
\address{Chenglang Yang, Institute for Math and AI, Wuhan University, Wuhan, China}
\address{Hua Loo-Keng Center for Mathematical Sciences,
	Academy of Mathematics and Systems Science,
	Chinese Academy of Sciences,
	Beijing, China}
\email{yangcl@pku.edu.cn}

\author{Jian Zhou}
\address{Jian Zhou, Department of Mathematical Sciences\\
	Tsinghua University, Beijing, China}
\email{jianzhou@mail.tsinghua.edu.cn}

\begin{abstract}
We present a mathematical proof of the Aganagic-Dijkgraaf-Klemm-Mari\~no-Vafa Conjecture proposed in 2006,
which states that the generating function of the topological vertex,
i.e., the generating function of the open Gromov-Witten invariants of $\mathbb{C}^3$,
satisfies the $3$-component KP hierarchy.
In our proof we introduce a boson-fermionic field assignment which generalizes the well-known boson-fermion correspondence.
The proof also works for the generalization to the framed topological vertex case conjectured by Deng and Zhou.
As a consequence,
open Gromov-Witten theory of all smooth toric Calabi-Yau threefolds are controlled by the multi-component KP hierarchy.
\end{abstract}

\maketitle


\section{Introduction}

Enumerative geometry and integrable systems are deeply related to each other
from the viewpoint of string theory.
This connection starts from the famous Witten Conjecture/Kontsevich Theorem \cite{wit90, kon}
which relates the intersection theory on the moduli spaces of stable curves to the KdV hierarchy.
Later this was generalized to the case of $r$-spin intersection numbers
and the $r$-KdV hierarchy \cite{wit2, fsz},
and the Fan-Jarvis-Ruan-Witten theory of ADE \cite{fjr}, BCFG \cite{lrz} types
and the corresponding Drinfeld-Sokolov integrable systems.
This connection also appears in the Gromov-Witten theory with a non-trivial target,
for example,
the Gromov-Witten invariants of the projective line $\bC\bP^1$
was shown to be related to the 2D Toda lattice hierarchy by Okounkov-Pandharipande \cite{op},
see also \cite{ey,dz}.
It is natural to expect that enumerative invariants of Gromov-Witten type
are all controlled by some infinite-dimensional integrable systems.
See \cite{dz01} and \cite{bur} for two constructions of integrable systems
from a cohomological field theory \cite{km}.

In this paper,
we settle the cases of the open Gromov-Witten theory of smooth toric Calabi-Yau threefolds
and establish their connection to the multi-component KP hierarchy.
This connection was conjectured by string theorists
Aganagic, Dijkgraaf, Klemm, Mari\~no and Vafa \cite{adkmv} in 2006.
Roughly speaking, the open Gromov-Witten invariants of a smooth toric Calabi-Yau threefold $X$
count holomorphic maps from Riemann surfaces with boundaries to $X$
such that the boundaries are mapped to some given Lagrangians in $X$.
Motivated by the duality between large $N$ Chern-Simons theory and the string theory \cite{wit},
Aganagic, Klemm,  Mari\~no and Vafa \cite{akmv} proposed a formalism called the topological vertex
which provides a powerful tool to compute the generating series of all-genera
open Gromov-Witten invariants of a toric Calabi-Yau threefold.
This generating series is called the open string amplitudes of corresponding models in physics literatures.
The topological vertex is the generating series of the open Gromov-Witten invariants of $\mathbb{C}^3$,
and the generating series of a general smooth toric Calabi-Yau threefold can be
obtained from the topological vertex by certain gluing rules.
The way to glue some topological vertices together
can be combinatorially read from the toric diagram associated with this threefold.
A mathematical theory of the topological vertex
has been established by Li, Liu, Liu and the third-named author \cite{lllz}
based on a sequence of earlier works on localizations on moduli spaces of relative stable maps.
Another mathematical approach to the topological vertex was
given by Maulik, Oblomkov, Okounkov, Pandharipande \cite{moop}
using the Gromov-Witten/Donaldson-Thomas correspondence.

The purpose of this work is to give a mathematical proof of a conjecture of Aganagic {\em et al} \cite{adkmv},
which represents the topological vertex as a Bogoliubov transform of the
vacuum vector in the $3$-component fermionic Fock space.
Let $\mu^1,\mu^2,\mu^3$ be three partitions of integers,
and let $W_{\mu^1,\mu^2,\mu^3}(q)$ be the topological vertex labeled by $\mu^1,\mu^2,\mu^3$.
Based on the duality between Chern-Simons theory and string theory,
Aganagic-Klemm-Mari\~no-Vafa \cite{akmv} represented the topological vertex
as an explicit but complicated formula in terms of skew Schur functions,
see \eqref{eqn:def W as skew}.
It is natural to consider the following generating function:
\be
\label{eqn:gf of top vert}
\sum_{\mu^1, \mu^2, \mu^3 \in\mathcal{P}} W_{\mu^1,\mu^2,\mu^3}(q) \; s_{\mu^1}(\mathbf{t}^1)
s_{\mu^2}(\mathbf{t}^2)
s_{\mu^3}(\mathbf{t}^3),
\ee
where $s_{\mu}(\mathbf{t})$ is the Schur polynomial
in formal variables $\mathbf{t}^i=(t^i_1,t^i_2,t^i_3,\cdots)$,
and $\cP$ is the set of all partitions of integers.
Motivated by the conjectural multi-component KP integrability of the generating function \eqref{eqn:gf of top vert},
one can apply the boson-fermion correspondence to $W_{\mu^1,\mu^2,\mu^3}(q)$
and represent it as a vector in the fermionic Fock space,
see \cite{akmv, adkmv}.
For an introduction of the boson-fermion correspondence and KP hierarchy,
see e.g. \cite{djm, sa, kv}.
The ADKMV Conjecture \cite{adkmv} is the following
conjectural fermionic representation of the topological vertex:
\begin{Conjecture}[ADKMV Conjecture]
The topological vertex is given by the following Bogoliubov transform
in the $3$-component fermionic Fock space $\cF \otimes \cF \otimes \cF$:
\be
\label{eqn:ADKMV}
W_{\mu^1,\mu^2,\mu^3}(q) =
 \langle \mu^1,\mu^2,\mu^3|
\exp\Big( \sum_{i,j =1,2,3} \sum_{m,n\geq 0} A_{mn}^{ij}(q)
\psi_{-m-\half}^{i} \psi_{-n-\half}^{j*} \Big)
|0\rangle \otimes |0\rangle \otimes |0\rangle,
\ee
where the coefficient $A_{mn}^{ij}(q)$ are given by the following combinatorial formulas:
\be\label{eqn:Aijmn}
\begin{split}
	&A_{mn}^{ii}(q) = (-1)^n \frac{q^{m(m+1)/4 - n(n+1)/4}}{[m+n+1]\cdot [m]![n]!},\\
	&A_{mn}^{i(i+1)} (q) = (-1)^n q^{m(m+1)/4 - n(n+1)/4 +1/6}
	\sum_{l=0}^{\min(m,n)} \frac{q^{(l+1)(m+n-l)/2}}{[m-l]![n-l]!},\\
	&A_{mn}^{i(i-1)} (q) = (-1)^{n+1} q^{m(m+1)/4 - n(n+1)/4 -1/6}
	\sum_{l=0}^{\min(m,n)} \frac{q^{-(l+1)(m+n-l)/2}}{[m-l]![n-l]!},\\
\end{split}
\ee
and we use the conventions $A_{mn}^{34} = A_{m,n}^{31}$, $A_{mn}^{10}=A_{mn}^{13}$,
and
\be
[m]! = \prod_{k=1}^m [k] = \prod_{k=1}^m (q^{k/2} - q^{-k/2}).
\ee
for an integer $m\geq 1$, and $[0]! =1$.
\end{Conjecture}

Here $ \langle \mu^1,\mu^2,\mu^3| =  \langle \mu^1| \otimes \langle \mu^2| \otimes \langle \mu^3|$
is a vector in the $3$-component dual fermionic Fock space.
For a partition $\mu$,
the vector $|\mu\rangle \in \cF$ corresponds to the Schur function $s_\mu$
under the boson-fermion correspondence,
and $\langle \mu|$ is the dual vector of $|\mu\rangle$.
See \S \ref{sec-Prelim} for the notations,
and \S \ref{sec:def nKP} for more about
the Bogoliubov transformation and multi-component KP integrability.

A straightforward consequence of this fermionic representation
is that the generating series \eqref{eqn:gf of top vert} of the topological vertex,
i.e.,
the generating series of all open Gromov-Witten invariants of $\bC^3$,
satisfies the $3$-component KP hierarchy.
Aganagic {\em et al} \cite{akmv,adkmv} also hope that this kind of integrability holds for the open Gromov-Witten theory of
a general toric Calabi-Yau threefold (after some modifications).
This conjecture plays a crucial role in understanding the local mirror symmetry and local mirror curves.
Later Kashani-Poor \cite{kas} studied the case of resolved conifold,
and with the assumption of the $4$-component KP integrability
he computed the coefficients in corresponding Bogoliubov transform.
Takasaki \cite{tak} proved the KP integrability of the generalized conifold,
and he asserted that the full theory of these models
should satisfy the multi-component KP hierarchy
even though he is unable to prove it (see \cite[\S 4.1]{tak}).
Based on some techniques concerning the quantum torus algebra,
Takasaki and Nakatsu \cite{tn} proved the $2$-component KP integrability
of the closed topological vertex model
(and the full theory is supposed to be $6$-component KP integrable
since the toric diagram has $6$ external legs).

In 2015,
Deng and the third-named author obtained some partial results of the ADKMV Conjecture in \cite{dz1}.
They first generalized the ADKMV Conjecture to the case of the framed topological vertex
and then gave a proof of the $2$-legged case
(i.e., the case that $\mu^3$ is the empty partition),
which immediately implies the
$2$-component KP integrability of the generating function of $2$-legged (framed) topological vertex.
The framing plays a crucial role in the gluing procedures of topological vertices.
This generalization of the ADKMV Conjecture is called the framed ADKMV Conjecture,
and the original ADKMV Conjecture can be recovered from
the framed ADKMV Conjecture by simply taking all framings to zero.
See \S \ref{sec:adkmv} for a brief review of this generalization.
The proof of the $2$-legged case of the framed ADKMV Conjecture given in \cite{dz1}
involves complicated combinatorial computations of skew Schur functions
and it seems very hard to generalize this method directly to
the most general case, i.e., the $3$-legged case.
In another work \cite{dz2},
Deng and the third-named author proved that the gluing procedure of topological vertices
preserves the multi-component KP integrability.
More precisely,
they proved that the gluing of Bogoliubov transforms of the fermionic vacuum
is still (a part of) a Bogoliubov transform of the fermionic vacuum,
and thus one automatically obtains that
the generating functions of open Gromov-Witten invariants of a smooth toric Calabi-Yau manifold
is a part of some tau-function of the multi-component KP hierarchy
provided the framed ADKMV Conjecture.
See Theorem \ref{thm:provided adkmv} for a precise description of this statement.

It is not easy to generalize
the methods used in the literatures \cite{adkmv,dz1,dz2,kas,tak,tn}
dealing with the integrality of the open Gromov-Witten theory in the cases of $1$-component and $2$-component KP hierarchies
to the $3$-legged case.
The difficulty lies in the following two folds.
First,
in the $1$-component and $2$-component case
the combinatorial expressions in terms of Schur and skew Schur functions
can be represented using $1$-component free fermions via the boson-fermion correspondence,
but $1$-component free fermions would not be sufficient when one studies
the $n$-component KP integrability for a general $n\geq3$.
Second,
it is not easy to translate the combinatorial expression of the topological vertex in terms of skew Schur functions
as a vector in the fermionic Fock space.
One crucial step towards in this question
is obtained by Okounkov, Reshetikhin and Vafa \cite{orv}.
In that work,
they proved the equivalence between the topological vertex and the generating function of
enumeration of plane partitions, up to a MacMahon factor,
which is surprisingly essential to obtain a much simpler formula for the topological vertex
as a vacuum expectation value in the fermionic Fock space,
see \cite[(3.17), (3.21)]{orv}.
However,
their formula is still not enough to derive the multi-component KP integrability
since the dependence of their formula on the second partition is difficult to handle.
Moreover,
the operators appearing in their vacuum expectation value representation
are actually the free bosons instead of the free fermions,
and Wick's theorem does not work in this situation.

In this paper,
we prove the framed ADKMV Conjecture using another strategy.
To overcome the difficulties mentioned above,
we propose a boson-fermionic assignment
which generalizes the well-known boson-fermion correspondence,
which gives a new fermionic representation of the framed topological vertex.
More precisely,
we represent the framed topological vertex as a vacuum expectation value
of some operator on the $1$-component fermionic Fock space,
such that the dependence on the second partition is hidden in the operator.
For more details of this construction,
see \S \ref{sec:bf assign}.
Our result completes the picture described in last several paragraphs,
and in particular,
we now know that the Gromov-Witten theory of all smooth toric Calabi-Yau manifolds
are controlled by the multi-component KP hierarchy in the sense of Theorem \ref{thm:provided adkmv}.

In the rest of the Introduction,
we briefly state the main results of this work.
For an integer $a\in\mathbb{Z}$ (which corresponds to the framing data)
and a partition $\mu\in\mathcal{P}$,
we assign a fermionic field $\Psi_{\mu}^{(a)}(q)$
which is an operator on the fermionic Fock space.
The definition of each operator $\Psi_{\mu}^{(a)}(q)$ involves
only finitely many fermionic fields of the form:
\begin{equation*}
\tGa_+^{(a)} (z) = f_1^{(a)}(z) \Gamma_-(z^{-1}) \psi^*(z^{-1}),
\qquad \tGa_-^{(a)} (z) = z\cdot f_2^{(a)}(z) \psi(z) \Gamma_+(z^{-1}),
\end{equation*}
see \S \ref{sec:bf assign} for details.
This construction is similar to
the way how a fermionic state corresponds to a bosonic state in the ordinary boson-fermion correspondence,
and thus is suitable to study the multi-component KP integrability.
First we have:
\begin{Theorem} \label{thm:Main1}
For a partition $\mu$ and an integer $a\in \bZ$,
define an operator $\Psi_\mu^{(a)}$ on the fermionic Fock space as in Definition \ref{def-boson-ferm}.
Then the framed topological vertex is given by:
\be
W_{\mu^1,\mu^2,\mu^3}^{(a_1,a_2,a_3)} (q) =
\langle \mu^1 | q^{a_1 K/2} \Psi_{\mu^2}^{(a_2)}(q) q^{-(a_3 +1)K /2} |(\mu^3)^t \rangle,
\ee
where $K$ is the cut-and-join operator on the fermionic Fock space.
\end{Theorem}

Using this new fermionic expression for the framed topological vertex,
we are able to derive a determinantal formula using Wick's Theorem
which enables us to identify the framed topological vertex
with the Bogoliubov transform conjectured in \cite{adkmv, dz1}.
In this way we prove:
\begin{Theorem} \label{thm:Main2}
The framed ADKMV Conjecture holds.
More precisely,
for an arbitrary framing $(a_1,a_2,a_3)\in \bZ^3$,
the framed topological vertex
\begin{equation*}
W_{\mu^1,\mu^2,\mu^3}^{(a_1,a_2,a_3)}(q)
= q^{a_1\kappa_{\mu^1}/2 + a_2\kappa_{\mu^2}/2 + a_3\kappa_{\mu^3}/2} \cdot
W_{\mu^1,\mu^2,\mu^3}(q),
\end{equation*}
equals to a Bogoliubov transform of the $3$-component fermionic vacuum whose coefficients are
\be\label{eqn:def Aqa}
A_{mn}^{ij} (q;\bm a) = q^{\frac{a_i m(m+1) - a_j n(n+1)}{2}}
\cdot A_{mn}^{ij} (q).
\ee
\end{Theorem}

A Bogoliubov transform of the above form
automatically corresponds to a tau-function of the multi-component KP hierarchy
under the boson-fermion correspondence,
see \cite{kv}, and see \S \ref{sec:def nKP} for a brief review.
Thus as a corollary of the above theorem, we get:
\begin{Theorem} \label{thm:Main3}
For every $a_1, a_2, a_3 \in \bZ$,
the generating function of the (framed) topological vertex
(i.e.,
the generating function of the open Gromov-Witten invariants of $\mathbb{C}^3$):
\be
Z^{(a_1,a_2,a_3)}(\bm t^1,\bm t^2,\bm t^3)
= \sum_{\mu^1, \mu^2, \mu^3} W^{(a_1,a_2,a_3)}_{\mu^1,\mu^2,\mu^3}
\cdot s_{\mu^1}(\bm t^1)  s_{\mu^2}(\bm t^2)  s_{\mu^3}(\bm t^3)
\ee
satisfies the $3$-component KP hierarchy.
\end{Theorem}

As introduced in the previous paragraphs,
the (framed) topological vertex is the building block of
Aganagic-Klemm-Mari\~no-Vafa's algorithm to compute the open Gromov-Witten invariants of
a general toric Calabi-Yau threefold \cite{akmv,lllz,moop}.
Now by combining the above Theorems with the fact that
gluing of topological vertices preserves the multi-component KP integrability \cite{dz1},
we conclude that
the generating series of the open Gromov-Witten invariants of an arbitrary smooth toric Calabi-Yau threefold
satisfies the multi-component KP hierarchy (in the sense of Theorem \ref{thm:provided adkmv}).

We arrange the rest of the paper in the following fashion.
First we recall some preliminary facts and notations on the boson-fermion correspondence
in \S \ref{sec-Prelim}.
The framed topological vertex,
the ADKMV Conjecture,
and the framed ADKMV Conjecture will be recalled in \S \ref{sec:Conjecture}.
At the end of \S \ref{sec:Conjecture}
we also give a brief sketch of our proof of the (framed) ADKMV Conjecture.
The first three steps are carried out in \S \ref{sec-ferm-topovert}.
In that section we first introduce the boson-fermionic field assignment $(\mu, a) \mapsto \Psi^{(a)}_\mu(q)$,
and use it to define a combinatorial (framed) topological vertex $C^{(a_1,a_2,a_3)}_{\mu^1, \mu^2, \mu^3}(q)$.
Next we derive a determinantal formula for  $C^{(a_1,a_2,a_3)}_{\mu^1, \mu^2, \mu^3}(q)$,
and then prove Theorem \ref{thm:Main1} by showing the equivalence between this combinatorial topological vertex
and the original topological vertex $W_{\mu^1,\mu^2,\mu^3}^{(a_1,a_2,a_3)} (q)$.
Finally we prove Theorem \ref{thm:Main2} in \S \ref{sec:ADKMV-proof} in two steps.
First we derive a determinantal formula for the Bogoliubov transform
in the right-hand side of the (framed) ADKMV Conjecture,
and then we show that this determinantal formula actually matches with
the determinantal formula for $C^{(a_1,a_2,a_3)}_{\mu^1, \mu^2, \mu^3}(q)$.

\section{Preliminaries of Boson-Fermion Correspondence}
\label{sec-Prelim}

In this section,
we give a brief review of the preliminaries we need,
including some basic knowledge of symmetric functions,
free fermions, vertex operators,
and boson-fermion correspondence.

\subsection{Schur functions and skew Schur functions}

First we recall the definitions of Schur functions and skew Schur functions.
We use \cite[\S I]{mac} for references.

A partition
$\lambda=(\lambda_1,\lambda_2,\cdots,\lambda_{l(\lambda)})$ of the integer $n\geq 0$
is a sequence of integers $\lambda_1\geq \cdots\geq \lambda_{l(\lambda)}> 0$,
satisfying $|\lambda|=\lambda_1+\cdots+\lambda_l = n$.
The number $l(\lambda)$ is called the length of this $\lambda$.
In particular,
the empty partition $\lambda = (\emptyset )$ is a partition of $0$.
There is a one-to-one correspondence between the set of partitions of integers
and the set of Young diagrams,
and the Young diagram corresponding to $\lambda=(\lambda_1,\cdots,\lambda_{l(\lambda)})$
consists of $|\lambda|$ boxes such that there are exactly $\lambda_i$ boxes in the $i$-th row.

Let $\lambda=(\lambda_1,\cdots,\lambda_{l(\lambda)})$ be a partition,
then its transpose $\lambda^t=(\lambda_1^t,\cdots,\lambda_m^t)$ is the partition defined by
$m=\lambda_1$ and $\lambda_j^t = |\{i| \lambda_i\geq j\}|$.
The Young diagram corresponding to $\lambda^t$
is obtained by flipping the Young diagram corresponding to $\lambda$ along the diagonal.
It is clear that $(\lambda^t)^t =\lambda$ for every partition $\lambda$.
The Frobenius notation of a partition $\lambda$ is defined to be:
\begin{equation*}
\lambda= (m_1,m_2,\cdots,m_k | n_1,n_2,\cdots,n_k),
\end{equation*}
where $k$ is the number of boxes on the diagonal of the corresponding Young diagram,
and $m_i = \lambda_i - i$, $n_i=\lambda_i^t - i$ for $i=1,2,\cdots,k$.

Now we recall the definition of the Schur functions $s_\lambda$ indexed by a partition $\lambda$.
First consider the partition $\lambda=(m|n) = (m+1,1^n)$ whose Young diagram is a hook,
and in this case the Schur function $s_{(m|n)}$ is defined by:
\begin{equation*}
s_{(m|n)}= h_{m+1}e_n - h_{m+2}e_{n-1} + \cdots
+ (-1)^n h_{m+n+1},
\end{equation*}
where $h_n$ and $e_n$ are the
complete symmetric function and elementary symmetric function of degree $n$ respectively.
For a general $\lambda$
whose Frobenius notation is given by $(m_1,\cdots,m_k|n_1,\cdots,n_k)$,
the Schur function $s_\lambda$ is defined by:
\begin{equation*}
s_\lambda  = \det (s_{(m_i|n_j)} )_{1\leq i,j\leq k}.
\end{equation*}
There is also an equivalent definition for $s_\lambda$:
\be
s_\lambda = \det(h_{\lambda_i-i+j})_{1\leq i,j\leq n}
= \det(e_{\lambda_i^t-i+j})_{1\leq i,j\leq m},
\ee
for $n\geq l(\lambda)$ and $m\geq l(\lambda^t)$.
The Schur function indexed by the empty partition is defined to be $s_{(\emptyset)} = 1$.
Denote by $\Lambda$ the space of all symmetric functions,
then Schur functions $\{s_\lambda\}_{\lambda}$ form a basis of  $\Lambda$.
Let $\{s_\lambda\}_{\lambda}$ be an orthonormal basis,
and in this way one obtains an inner product $(\cdot,\cdot)$ on $\Lambda$.

Now let $\lambda,\mu$ be two partitions,
then the skew Schur function $s_{\lambda/\mu} \in \Lambda$ is defined by the following property:
\be
(s_{\lambda/\mu},s_\nu ) = (s_\lambda, s_\mu s_\nu),
\qquad \forall \text{ partition $\nu$}.
\ee
Or equivalently,
\begin{equation*}
s_{\lambda/\mu} = \sum_\nu c_{\mu\nu}^\lambda s_\nu
\end{equation*}
where $\{c_{\mu\nu}^\lambda\}$ are the Littlewood-Richardson coefficients:
\begin{equation*}
s_\mu s_\nu = \sum_{\lambda} c_{\mu\nu}^\lambda s_\lambda,
\end{equation*}
One has $s_{\lambda/\mu} = 0$ unless $\mu \subset \lambda$
(i.e., $\mu_i\leq \lambda_i$ for every $i$).
Skew Schur functions can also be represented as determinants
(cf. \cite[\S I.(5.4)]{mac}):
\begin{equation*}
s_{\lambda / \mu} = \det(h_{\lambda_i -\mu_j -i+j})_{1\leq i,j\leq n}
= \det(e_{\lambda_i^t -\mu_j^t -i+j})_{1\leq i,j\leq m} ,
\end{equation*}
for $n\geq l(\lambda)$ and $m\geq l(\lambda^t)$.
From this determinantal formula one knows that:
\be
\label{eq-skewS-length1}
s_{(m)/(n)} = h_{m-n} = s_{(m-n)},
\qquad
s_{(1^m)/(1^n)} = e_{m-n} = s_{(1^{m-n})},
\ee
for $m\geq n\geq 0$.

Now we consider a particular specialization of symmetric functions.
Let $q$ be a formal variable and $\rho = (-\half, -\frac{3}{2},-\frac{5}{2},\cdots)$,
and denote by $q^\rho$ the sequence $q^\rho = (q^{-1/2},q^{-3/2}, q^{-5/2},\cdots)$,
then it is easy to see that:
\begin{equation}\label{eqn:p_n at q^rho}
p_n (q^\rho) = \frac{1}{q^{n/2} - q^{-n/2}} = \frac{1}{[n]},
\end{equation}
where $p_n (q^\rho)$ means taking the evaluation $x_n = q^{-(2n-1)/2}$
in the Newton symmetric function $p_n (x_1,x_2,\cdots)$,
and we use the notation
\begin{equation*}
[n] = q^{n/2} - q^{-n/2}.
\end{equation*}
The following specialization of Schur functions will be useful
(\cite[p. 45, Example 2]{mac}, see also \cite[\S 4]{zhou4}):
\be
\label{eq-eval-schur-rho}
s_\mu (q^\rho) = q^{\kappa_\mu /4}\cdot \frac{1}{\prod_{e\in \mu} [h(e)]},
\ee
where for a partition $\mu = (m_1,\cdots, m_k|n_1,\cdots, n_k)$,
the number $\kappa_\mu$ is:
\be
\label{eq-def-kappamu}
\kappa_\mu = \sum_{i=1}^{l(\mu)} \mu_i(\mu_i -2i +1)
=\sum_{i=1}^k (m_i+\half)^2 - \sum_{i=1}^k (n_i+\half)^2,
\ee
in particular $\kappa_{\emptyset} =0$;
and $e\in \mu$ runs over all boxes in the Young diagram corresponding to $\mu$,
where $h(e)$ denotes the hook length of the box $e$:
\begin{equation*}
h(e) = (\lambda_i - i) + (\lambda_j^t -j) +1,
\end{equation*}
if $e$ is in the $i$-th row and $j$-th column of the Young diagram.

\subsection{Fermionic Fock space}

In this subsection we recall the free fermions and
the semi-infinite wedge construction of the fermionic Fock space.
See \cite{djm, kac, sa}.

Let $\bm{a}=(a_1,a_2,\cdots)$ be a sequence of half-integers $ a_i \in \bZ+\half$ such that $a_1<a_2<a_3<\cdots$.
We say $\bm{a}$ is admissible if:
\begin{equation*}
\big|(\bZ_{\geq 0}+\half)-\{a_1,a_2\cdots\}\big|<\infty,
\qquad
\big|\{a_1,a_2\cdots\}-(\bZ_{\geq 0}+\half)\big|<\infty.
\end{equation*}
Given an admissible $\bm{a}$,
one can associate a semi-infinite wedge product $|\bm a\rangle$:
\be
| \bm a\rangle =
z^{a_1} \wedge z^{a_2} \wedge z^{a_3} \wedge \cdots
\ee
The fermionic Fock space $\cF$ is the infinite dimensional vector space of
all formal (infinite) summations of the following form:
\begin{equation*}
\sum_{\bm a:\text{ admissible}} c_{\bm a} |\bm a\rangle.
\end{equation*}

Let $\bm a$ be admissible,
then the charge of the vector $|\bm a\rangle \in \cF$ is defined to be:
\begin{equation*}
\text{charge}(|\bm a\rangle)=
\big|\{a_1,a_2\cdots\}-(\bZ_{\geq 0}+\half)\big|
-\big|(\bZ_{\geq 0}+\half)-\{a_1,a_2\cdots\}\big|,
\end{equation*}
and this gives a decomposition of the fermionic Fock space:
\be
\cF=\bigoplus_{n\in \bZ} \cF^{(n)},
\ee
where $\cF^{(n)}$ is spanned by vectors of charge $n$.
We will denote:
\be
|n\rangle = z^{n+\half}\wedge z^{n+\frac{3}{2}}
\wedge z^{n+\frac{5}{2}} \wedge \cdots \in \cF^{(-n)}.
\ee
The vector $|0\rangle = z^{\half}\wedge z^{\frac{3}{2}}
\wedge z^{\frac{5}{2}} \wedge \cdots \in \cF^{(0)}$
is called the fermionic vacuum vector.

The subspace $\cF^{(0)} \subset \cF$ has a natural basis indexed by partitions of integers.
Let $\mu=(\mu_1,\mu_2,\cdots)$ be a partition where
$\mu_1 \geq \mu_2 \geq \cdots\geq \mu_l >\mu_{l+1}=\mu_{l+2}=\cdots=0$,
and denote:
\be
\label{eq-basisF-mu}
|\mu\rangle =
z^{\frac{1}{2}-\mu_1}\wedge z^{\frac{3}{2}-\mu_2}\wedge
z^{\frac{5}{2}-\mu_3}\wedge \cdots \in\cF^{(0)},
\ee
then $\{|\mu\rangle\}_\mu$ form a basis for $\cF^{(0)}$.
Here the empty partition $(0,0,\cdots)$ of $0\in \bZ$
corresponds to the vacuum vector $|0\rangle\in\cF^{(0)}$.

The free fermions $\psi_r, \psi_r^*$ ($r\in \bZ+\half$) are a family of operators satisfying the following anti-commutation relations:
\be
\label{eq-anticomm-psi}
[\psi_r,\psi_s]_+=0,
\qquad
[\psi_r^*,\psi_s^*]_+=0,
\qquad
[\psi_r,\psi_s^*]_+= \delta_{r+s,0}\cdot \Id,
\ee
where the bracket is defined by $[\phi,\psi]_+ =\phi\psi+\psi\phi$.
These fermions generate a Clifford algebra.
The Clifford algebra acts on the fermionic Fock space $\cF$ by:
\ben
\psi_r |\bm a\rangle = z^r \wedge |\bm a\rangle,
\qquad \forall r\in \bZ+\half,
\een
and
\ben
\psi_r^* | \bm a\rangle=
\begin{cases}
(-1)^{k+1} \cdot z^{a_1}\wedge z^{a_2}\wedge \cdots \wedge \widehat{z^{a_k}} \wedge \cdots,
&\text{if $a_k = -r$ for some $k$;}\\
0, &\text{otherwise.}
\end{cases}
\een
It is clear that $\{\psi_r\}$ all have charge $1$,
and $\{\psi_r^*\}$ all have charge $-1$.
The operators $\{\psi_r,\psi_r^*\}_{r<0}$ are called the fermionic creators,
and $\{\psi_r,\psi_r^*\}_{r>0}$ are called the fermionic annihilators.
One can easily check that:
\be
\psi_r |0\rangle=0,
\qquad
\psi_r^* |0\rangle=0,
\qquad \forall r>0.
\ee
Moreover,
every element of the form $|\mu\rangle$ (where $\mu$ is a partition)
can be obtained by applying fermionic creators successively to the vacuum $|0\rangle$:
\be
\label{eq-ferm-basismu}
|\mu\rangle=(-1)^{n_1+\cdots+n_k}\cdot
\psi_{-m_1-\half} \psi_{-n_1-\half}^* \cdots
\psi_{-m_k-\half} \psi_{-n_k-\half}^* |0\rangle,
\ee
if the Frobenius notation for $\mu=(\mu_1,\mu_2,\cdots)$ is
$\mu=(m_1,\cdots, m_k | n_1,\cdots,n_k)$.

One can define an inner product $(\cdot,\cdot)$ on $\cF$ by taking
$\{|\bm a\rangle | \text{$\bm a$ is admissible}\}$ to be an orthonormal basis,
and then $\{|\mu  \rangle \}_\mu$ is an orthonormal basis for $\cF^{(0)}$.
Given two admissible sequences $\bm a$ and $\bm b$,
we denote by $\langle \bm b | \bm a\rangle =(|\bm a\rangle,|\bm b\rangle)$
the inner product of the two vectors $|\bm a\rangle , |\bm b\rangle \in \cF$.
It is clear that
$\psi_r$ and $\psi_{-r}^*$ are adjoint to each other with respect to this inner product.

Let $A$ be an operator on the Fock space $\cF$,
and we will denote by
\begin{equation*}
\langle A \rangle = \langle 0| A |0\rangle
\end{equation*}
the vacuum expectation value of $A$.
Let $w_1, w_2,\cdots,w_k$ be some linear combinations of the free fermions $\{\psi_r,\psi_r^*\}$,
then the vacuum expectation value of the product $w_1w_2\cdots w_k$ can be computed using Wick's Theorem
(see e.g. \cite[\S 4.5]{djm}):
\be
\label{eq-Wickthm}
\langle  w_1 w_2 \cdots w_k \rangle
=\begin{cases}
0, & \text{ if $k$ is odd;}\\
\sum\limits_\sigma (-1)^\sigma  \cdot \prod\limits_{j=1}^{k/2} \langle w_{\sigma(2j-1)} w_{\sigma(2j)} \rangle,
& \text{ if $k$ is even,}\\
\end{cases}
\ee
where the summation is over permutations $\sigma \in S_{k}$ such that
$\sigma(2j-1)<\sigma(2j)$ for every $j$ and
$\sigma(1) < \sigma(3) <\sigma(5) <\cdots <\sigma(k-1)$;
and $(-1)^\sigma = \pm 1$ denotes the sign of this permutation.
And the vacuum expectation value of the product of two free fermions are given by:
\be
\label{eq-quadVEV}
\begin{split}
&\langle \psi_r^* \psi_s^* \rangle = \langle \psi_r \psi_s \rangle =0,
\qquad \forall r,s\in \bZ+\half;\\
&\langle \psi_r^* \psi_s \rangle = \langle \psi_r \psi_s^* \rangle
=\begin{cases}
1,&\text{ if $r=-s>0;$}\\
0,&\text{ otherwise.}
\end{cases}
\end{split}
\ee

One can also rewrite Wick's Theorem via determinant.
Let $\varphi_1,\varphi_2,\cdots,\varphi_n$ be
some linear summations of $\{\psi_r\}_{r\in \bZ+\half}$,
and let $\varphi_1^*,\varphi_2^*,\cdots,\varphi_n^*$ be
some linear summations of $\{\psi_r^*\}_{r\in \bZ+\half}$.
Then one has:
\be
\label{eq-Wick-det-1}
\langle \varphi_1 \varphi_1^* \varphi_2 \varphi_2^* \cdots \varphi_n \varphi_n^* \rangle
= \sum_{\bm p} (-1)^{\bm p} \cdot M(\bm p)_1 M(\bm p)_2 \cdots M(\bm p)_n,
\ee
where the sequence $\bm p = (p_1,\cdots ,p_n)$ runs over permutations of $(1, 2, \cdots , n)$,
and $M(\bm p)_i$ is:
\begin{equation*}
M(\bm p)_i = \begin{cases}
\langle \varphi_i \varphi_{p_i}^* \rangle, &\text{ if $p_i\geq i$;}\\
- \langle  \varphi_{p_i}^*  \varphi_i \rangle, &\text{ if $p_i < i$.}\\
\end{cases}
\end{equation*}
Clearly the right-hand side of \eqref{eq-Wick-det-1} is a determinant,
and thus one concludes that:
\be
\label{eq-Wick-det-2}
\langle \varphi_1 \varphi_1^* \varphi_2 \varphi_2^* \cdots \varphi_n \varphi_n^* \rangle
= \det (M_{ij})_{1\leq i,j\leq n},
\ee
where the entries of the $n\times n$ matrix $M$ is:
\begin{equation*}
M_{ij} = \begin{cases}
\langle \varphi_i \varphi_{j}^* \rangle, &\text{ if $j\geq i$;}\\
- \langle  \varphi_{j}^*  \varphi_i \rangle, &\text{ if $j < i$.}\\
\end{cases}
\end{equation*}

\subsection{Cut-and-join operator}
In this subsection,
we recall the cut-and-join operator on the fermionic Fock space \cite{Ok1, gjv}.

The cut-and-join operator $K$ is defined to be:
\be
\label{eq-def-C&Jopr}
K = \sum_{s\in \bZ+\half} s^2
:\psi_s \psi_{-s}^*:.
\ee
Then $K$ is self-adjoint.
Moreover, one can easily check that:
\be
\label{eq-comm-Kpsi}
[K,\psi_r]= r^2 \psi_r,
\qquad
[K,\psi_r^*] = - r^2 \psi_r^*.
\ee
Then by the Baker-Campbell-Hausdorff formula
\be
\label{eq-BCH}
e^{-X} Y e^X = Y -[X,Y] + \frac{1}{2!}[X,[X,Y]] + \frac{1}{3!} [X, [X,[X,Y]]] -\cdots,
\ee
one has:
\be
\label{eq-comm-expKpsi}
e^{-K} \psi_r e^K = e^{-r^2}\psi_r,
\qquad
e^{-K} \psi_r^* e^K = e^{r^2}\psi_r^*.
\ee
It is well-known that
the basis vectors $|\mu\rangle$ are eigenvectors of the cut-and-join operator $K$:
\be
\label{eq-eigen-C&J}
K |\mu \rangle  =  \kappa_\mu |\mu\rangle,
\ee
where the corresponding eigenvalue $\kappa_\mu$ is the number \eqref{eq-def-kappamu}.
This fact can be easily seen by using \eqref{eq-ferm-basismu} and the commutation relations \eqref{eq-comm-Kpsi}.

\subsection{Boson-fermion correspondence}
\label{sec-b-f}

In this subsection,
we recall the bosonic Fock space and boson-fermion correspondence.
See \cite{djm} for details.

Let $\alpha_n$ be the following operator on $\cF$:
\be
\alpha_n = \sum_{s\in \bZ+\half} :\psi_{-s} \psi_{s+n}^*:.
\qquad n\in \bZ,
\ee
Here $: \psi_{-s} \psi_{s+n}^* :$ denotes the normal-ordered product of fermions:
\begin{equation*}
:\phi_{r_1}\phi_{r_2}\cdots \phi_{r_n}:
=(-1)^\sigma \phi_{r_{\sigma(1)}}\phi_{r_{\sigma(2)}}\cdots\phi_{r_{\sigma(n)}},
\end{equation*}
where $\phi_k$ is either $\psi_k$ or $\psi_k^*$,
and $\sigma\in S_n$ is a permutation such that $r_{\sigma(1)}\leq\cdots\leq r_{\sigma(n)}$.
The operators $\{\alpha_n\}_{n\in \bZ}$ are called the bosons.
Denote the generating series of free fermions by:
\be
\label{eq-gen-fermi}
\psi(\xi)= \sum_{s\in\bZ+\half} \psi_s \xi^{-s-\half},
\qquad
\psi^*(\xi)= \sum_{s\in\bZ+\half} \psi_s^* \xi^{-s-\half},
\ee
then the generating series of the bosons is:
\be
\label{eq-gen-bos}
\alpha(\xi)= \sum_{n\in\bZ} \alpha_n \xi^{-n-1} :\psi(\xi)\psi^*(\xi): .
\ee
These bosons  satisfy the following commutation relations:
\be
\label{eq-comm-boson}
[\alpha_m,\alpha_n]= m\delta_{m+n,0} \cdot \Id,
\ee
i.e.,
they generate a Heisenberg algebra.
One can also check that:
\be
\label{eq-comm-alphapsi}
[\alpha_n, \psi_r] = \psi_{n+r},
\qquad
[\alpha_n, \psi_r^*] = -\psi_{n+r}^*.
\ee
The operator $\alpha_0$ is called the charge operator, and will be denoted by $C$:
\be
C = \sum_{s\in \bZ+\half} :\psi_{-s} \psi_{s}^*:.
\ee
Then $\cF^{(n)}$ is the eigenspace of $C$
with respect to the eigenvalue $n\in \bZ$.

The bosonic Fock space $\cB$ is defined by $\cB:=\Lambda[[w,w^{-1}]]$,
where $\Lambda$ is the space of symmetric functions in some formal variables
$\bm x=(x_1,x_2,\cdots)$,
and $w$ is a formal variable.
The boson-fermion correspondence is a linear isomorphism $\Phi:\cF \to \cB$ of vector spaces,
given by (see e.g. \cite[\S 5]{djm}):
\be
\Phi:\quad
|\bm a\rangle \in \cF^{(m)}
\quad\mapsto\quad
w^m\cdot \langle m |
e^{\sum_{n=1}^\infty \frac{p_n}{n} \alpha_n}
| \bm a \rangle,
\ee
where $p_n \in \Lambda$ is the Newton symmetric function of degree $n$.
In particular,
by restricting to $\cF^{(0)}$ one obtains an isomorphism:
\be
\label{eq-b-f-corresp}
\cF^{(0)}\to \Lambda,
\qquad
|\mu\rangle \mapsto s_\mu =
\langle 0 | e^{\sum_{n=1}^\infty \frac{p_n}{n} \alpha_n} | \mu \rangle,
\ee
where $|\mu\rangle \in \cF^{(0)}$ is the basis vector \eqref{eq-basisF-mu},
and  $s_\mu$ is the Schur function indexed by $\mu$.

\subsection{Properties of the vertex operators $\Gamma_\pm$}

Let $\bm t =(t_1,t_2,t_3,\cdots)$ be a sequence of formal variables,
and define $\Gamma_\pm$ to be the following operators on $\cF$:
\be
\label{e-def-gammat}
\Gamma_\pm (\bm t) = \exp\Big( \sum_{n=1}^\infty t_n\alpha_{\pm n} \Big).
\ee
Let $z$ be a formal variable,
and let $\{z\}$ be the sequence $\{z\} = (z,\frac{z^2}{2},\frac{z^3}{3},\cdots)$.
Then:
\be
\label{def-vert-Gamma}
\Gamma_\pm ( \{z\} ) = \exp\Big( \sum_{n=1}^\infty \frac{z^n}{n} \alpha_{\pm n} \Big).
\ee
Since $[\alpha_m,\alpha_n]= m\delta_{m+n,0} \cdot \Id$ is central,
one has:
\be
\label{eq-comm-gammapm}
\begin{split}
\Gamma_+ (\{z\}) \Gamma_- (\{w\}) = &
\exp\Big(\big[\sum_{n\geq 1}\frac{z^n}{n}\alpha_n, \sum_{m\geq 1}\frac{w^m}{m}\alpha_{-m} \big] \Big)
\Gamma_- (\{w\})\Gamma_+ (\{z\})\\
=& \frac{1}{1-zw}  \Gamma_- (\{w\})\Gamma_+ (\{z\}).
\end{split}
\ee
The fermionic fields $\psi(z),\psi^*(z)$ can be recovered from $\Gamma_{\pm}$ by
(cf. \cite[Ch.14]{kac}):
\be
\label{eq-psi-inGamma}
\begin{split}
&\psi(z) = z^{C-1} R \Gamma_- (\{z\}) \Gamma_+ (-\{z^{-1}\}),\\
&\psi^*(z) = R^{-1} z^{-C} \Gamma_- (-\{z\}) \Gamma_+ (\{z^{-1}\}),
\end{split}
\ee
where $C$ is the charge operator,
and $R$ is the shift operator on $\cF$ defined by:
\begin{equation*}
R (z^{a_1} \wedge z^{a_2} \wedge z^{a_3} \wedge \cdots ) =
z^{a_1-1} \wedge z^{a_2-1} \wedge z^{a_3-1} \wedge \cdots.
\end{equation*}
(Here \eqref{eq-psi-inGamma} looks different from \cite[Thm 14.10]{kac}
since we use different notations.)

Take $t_n = \frac{p_n}{n} \in \Lambda$ for every $n\geq 1$,
then $\Gamma_- (\bm t)$ act on $\cF$ by
(cf. \cite[(A.15)]{Ok2}):
\be
\Gamma_-(\bm t) |\mu\rangle = \sum_\lambda s_{\lambda/\mu}|\lambda\rangle.
\ee
It is clear that $\Gamma_+(\bm t)$ and $\Gamma_-(\bm t)$ are adjoint to each other
with respect to the inner product $(\cdot,\cdot)$ on $\cF$,
and thus:
\begin{equation*}
(\Gamma_+(\bm t)|\mu\rangle , |\nu\rangle)
= (|\mu\rangle , \Gamma_-(\bm t)|\nu\rangle)
=(|\mu\rangle, \sum_\lambda s_{\lambda/\nu}|\nu\rangle) =s_{\mu/\nu},
\end{equation*}
and thus:
\be
\label{eq-special-schur}
\Gamma_+(\bm t) |\mu\rangle = \sum_\lambda s_{\mu/\lambda}|\lambda\rangle.
\ee
And by taking dual one obtains:
\be
\label{eq-special-schur-2}
\langle \mu | \Gamma_+(\bm t) =\sum_\lambda s_{\lambda/\mu} \langle\lambda|,
\qquad
\langle \mu | \Gamma_-(\bm t) =\sum_\lambda s_{\mu/\lambda} \langle\lambda|.
\ee

\subsection{KP hierarchy, Bogoliubov transform, and multi-component generalization}
\label{sec:def nKP}

The $N$-component fermionic Fock space is the tensor product $\cF_1\otimes \cF_2 \otimes \cdots \otimes \cF_N$,
where $\cF_1, \cF_2,\cdots,\cF_N$ are $N$ copies of the fermionic Fock space $\cF$.
One can define the action of fermions $\{\psi_r^i,\psi_r^{i*}\}$ for $r\in \bZ+\half$ and $i=1,2,\cdots,N$
on the $N$-component fermionic Fock space
such that $\psi_r^i,\psi_r^{i*}$ act on the $i$-th component $\cF_i$
in the same way as the action of $\psi_r,\psi_r^{*}$ on $\cF$.
These fermions satisfy the anti-commutation relations:
\be
\label{eqn:def multi fermions}
\begin{split}
[\psi_r^i,\psi_s^j]  =  [\psi_r^{i*},\psi_s^{j*}] =0,
\qquad
[\psi_r^i,\psi_s^{j*}] = \delta_{r+s,0} \cdot \delta_{i,j} \cdot \Id.
\end{split}
\ee
for every $r,s\in\bZ+\half$ and $1\leq i,j\leq N$.

A Bogoliubov transform (of the vacuum vector) in the Fock space $\cF_1\otimes \cdots \otimes \cF_N$
is a vector of the following form:
\begin{equation}\label{eqn:def V}
|V\rangle= \exp \Big( \sum_{i,j=1}^N \sum_{m,n\geq 0} A_{mn}^{ij} \psi_{-m-\half}^i \psi_{-n-\half}^{j*} \Big)
|0\rangle\otimes \cdots \otimes |0\rangle.
\end{equation}
Notice that the operators in the above equation only involve fermionic creators.
The Bogoliubov transform of the vacuum vector is deeply related to the KP and multi-component KP hierarchy.
Roughly speaking, under the boson-fermion correspondence,
such a Bogoliubov transform corresponds to a tau-function of the $N$-component KP hierarchy.
For details we recommend the references \cite{djm,kv}.

Below we review a precise description of the multi-component KP hierarchy.
First for the case of $N=1$,
we use the following definition of a tau-function of the KP hierarchy
in terms of the Hirota bilinear relation:
\begin{Definition}
	\label{def: KP}
	The formal power series $\tau(\mathbf{t})\in\mathbb{C}[\![t_1,t_2,\cdots]\!]$ is a tau-function of the KP hierarchy if
the following Hirota bilinear relation holds:
	\begin{align*}
		\oint_{\infty} dz
		\quad e^{\xi(\mathbf{t}-\mathbf{t}',z)}
		\tau(\mathbf{t}-[z^{-1}])
		\tau(\mathbf{t}'+[z^{-1}])
		=0,
	\end{align*}
where $\mathbf{t}=(t_1,t_2,\cdots)$, and
	\begin{align*}
		\mathbf{t}\pm[z^{-1}]
		=\{t_1&\pm z^{-1}, t_2\pm \frac{z^{-2}}{2},
		t_3\pm \frac{z^{-3}}{3},\cdots\},
\qquad
		\xi(\mathbf{t},z)=\sum_{k\geq0}t_k z^k.
	\end{align*}
\end{Definition}
Denote $u=2\frac{\partial^2}{\partial t_1^2} \log \tau(\mathbf{t})$,
then the first constraint of tau-functions of the KP hierarchy in the above definition is
the so-called KP equation:
\begin{align*}
	\frac{3}{4}\frac{\partial^2 u}{\partial t_2^2}
	=\frac{\partial}{\partial t_1}
	\Big(\frac{\partial u}{\partial t_3}
	-\frac{3}{2}u\frac{\partial u}{\partial t_1}
	-\frac{1}{4}\frac{\partial^3 u}{\partial t_1^3}\Big),
\end{align*}
which is a generalization of the famous KdV equation.
Now let $|V_1\rangle$ be a Bogoliubov transform of the form:
\begin{equation}
	\label{eqn:def v1}
	|V_1\rangle= \exp \Big(\sum_{m,n\geq 0} A_{mn} \psi_{-m-\half} \psi_{-n-\half}^{*} \Big)
	|0\rangle
\end{equation}
for some coefficients $A_{mn}\in\mathbb{C}$ where $m,n\in\mathbb{Z}_{\geq0}$,
and it is well-known that
the formal series
\begin{align}
	\label{eqn:def tauV1}
	\tau_{V_1}(\mathbf{t})
	=\langle 0| \Gamma_+(\mathbf{t}) |V_1\rangle
\end{align}
is a tau-function of the KP hierarchy
where the operator $\Gamma_+(\mathbf{t})$ is defined by \eqref{e-def-gammat},
see e.g. \cite[\S 9.2]{djm}.
The correspondence between the vector $|V_1\rangle \in \cF$ and the formal power series $\tau_{V_1}(\mathbf{t})$
via \eqref{eqn:def tauV1} is exactly the boson-fermion correspondence reviewed in \S \ref{sec-b-f}.

The above discussions can be directly generalized to the case of a general $N$
when one considers the $N$-component fermions defined by \eqref{eqn:def multi fermions}
and the $N$-component boson-fermion correspondence.
We omit the details and recommend the paper of Kac and van de Leur \cite{kv} for details.
In particular,
a formal power series $\tau(\mathbf{t}^1,\mathbf{t}^2,\mathbf{t}^3)\in\mathbb{C}[\![t^j_k|j=1,2,3; k=1,2,\cdots]\!]$ is a tau-function of the $3$-component KP hierarchy if (see \cite[\S 2.3]{kv})
\begin{align}\label{eqn:multi-bi}
	\oint_{\infty} dz
	\quad\sum_{j=1}^3
	e^{\sum_{k\geq0}\big(t_k^j-s_k^{j}\big) z^k}
	\tau(\mathbf{t}^j-[z^{-1}])
	\tau(\mathbf{s}^{j}+[z^{-1}])
	=0,
\end{align}
where $\mathbf{t}^j=(t^j_k)_{k=1,2,\cdots}$ and
$\mathbf{s}^j=(s^j_k)_{k=1,2,\cdots}$ for $j=1,2,3$.
Then,
under the $N$-component boson-fermion correspondence,
the formal power series $\tau_V(\mathbf{t}^1,\mathbf{t}^2,\mathbf{t}^3)$
corresponding the Bogoliubov transform  \eqref{eqn:def V} satisfies
the multi-component bilinear relation \eqref{eqn:multi-bi},
thus is a tau-function of $N$-component KP hierarchy.
Such a tau-function $\tau_V(\mathbf{t}^1,\mathbf{t}^2,\mathbf{t}^3)$ can be represented
in terms of the Schur basis as:
\begin{align*}
	\tau_V(\mathbf{t}^1,\mathbf{t}^2,\mathbf{t}^3)
	=\sum_{\mu^1, \mu^2, \mu^3} c_{\mu^1,\mu^2,\mu^3} \cdot s_{\mu^1}(\bm t^1)  s_{\mu^2}(\bm t^2)  s_{\mu^3}(\bm t^3),
\end{align*}
where the coefficients $c_{\mu^1,\mu^2,\mu^3}$ are given by
\begin{align*}
	c_{\mu^1,\mu^2,\mu^3}
	=(|\mu^1\rangle\otimes|\mu^2\rangle\otimes|\mu^3\rangle,|V\rangle)
	=\langle \mu^1, \mu^2 , \mu^3|V\rangle.
\end{align*}

\begin{Remark}
	In \cite[\S 2.2]{kv},
Kac and van de Leur's
 definition of a tau-function of the $N$-component KP hierarchy
 is actually a family of functions,
  depending on elements of the root lattice of $sl_n(\mathbb{C})$.
  Given a Bogoliubov transform $|V\rangle$ of the form \eqref{eqn:def V},
	one actually obtains a tau-function in the sense of \cite{kv}
by applying the multi-component boson-fermion correspondence
which collects the contributions from all charges in different components.
Now in this paper,
we will not need the whole family of functions in  such a tau-function
since we are only concerned with the part that comes from the open Gromov-Witten theory.
Here we only consider the contribution of charge zero,
and in this way we only obtain a single formal power series $\tau(\mathbf{t}^1,\mathbf{t}^2,\mathbf{t}^3)\in\mathbb{C}[\![t^j_k]\!]$.
In this paper we will simply call this series a tau-function of the multi-component KP hierarchy.
\end{Remark}

\section{Framed Topological Vertex and Framed ADKMV Conjecture}
\label{sec:Conjecture}

In this section we first recall the open Gromov-Witten theory of $\mathbb{C}^3$ and its relation to the framed topological vertex.
The framed ADKMV Conjecture proposes a fermionic expression of the framed topological vertex as a Bogoliubov transform.
We claim that this conjecture is true,
and as a corollary,
the partition function of the open Gromov-Witten invariants of $\mathbb{C}^3$ is a tau-function of the $3$-component KP hierarchy.
We will give an outline of our proof in \S \ref{sec-outline}.
The details will be given in the next two sections.

\subsection{Partition function of open Gromov-Witten invariants of $\mathbb{C}^3$}

The topological vertex was first introduced by Aganagic {\em et al} \cite{akmv} from a string theoretical approach.
In the mathematical literature,
the mathematical theory of the topological vertex was developed in the setting of
 the open Gromov-Witten theory of $\mathbb{C}^3$ \cite{lllz}.
In this subsection we briefly recall the mathematical approach.

Li, Liu, Liu and the third-named author \cite{lllz} used the localization technique to study the local relative Gromov-Witten theory of toric Calabi-Yau threefolds.
By virtual localization formula,
the local relative Gromov-Witten invariants in this case can be expressed in terms of certain triple Hodge integrals.

Denote by $\overline{\mathcal{M}}_{g,n}$ the moduli spaces of stable curves of genus $g$ with $n$ marked points.
Let $\psi_i=c_1(\mathcal{L}_i)$ (where $i=1,\cdots,n$) be the first Chern class of the $i$-cotangent line bundle,
and $\lambda_j=c_j(\mathbb{E})$ (where $j=1,\cdots,g$) be the $j-$th Chern class of the Hodge bundle.
The triple Hodge integral considered in \cite{lllz} is defined as follows.
For given three partitions $\mu^1, \mu^2, \mu^3$,
and three numbers $\bm w =(w_1,w_2,w_3)$ satisfying $w_1+w_2+w_3=0$,
define
\begin{align}\label{eqn:Gg}
	\begin{split}
	G_{g,\mu^1,\mu^2,\mu^3}(\bm w) = &\frac{(-\sqrt{-1})^{l} }{\prod_{j=1}^3 |\Aut(\mu^j)|}
	\prod_{i=1}^3\prod_{j=1}^{\ell(\mu^i)}
	\frac{\prod_{a=1}^{\mu^i_j-1}(\mu^i_j w_{i+1} + a w_i) }
	{(\mu^i_j-1)!w_i^{\mu^i_j-1} }\\
	&\quad\quad\quad\cdot\int_{\Mbar_{g,l}
	} \prod_{i=1}^3\frac{\Lambda_g^\vee(w_i)w_i^{l-1} }
	{\prod_{j=1}^{l(\mu^i)}\big(w_i(w_i-\mu^i_j\psi_{d^i+j})\big)},
	\end{split}
\end{align}
where $\Lambda_g^\vee(u)=u^g-\lambda_1 u^{g-1}+\cdots + (-1)^g\lambda_g$,
and
\begin{align*}
	d^1=0,
	\quad d^2=l(\mu^1),
	\quad d^3=l(\mu^1)+l(\mu^2),
	\quad l=l(\mu^1)+l(\mu^2)+l(\mu^3).
\end{align*}
Here use the convention $w_4=w_1$.
The number of automorphisms of a partition $\mu$
is defined to be
$|\Aut(\mu)|=\prod_{i\geq1} m_i(\mu)!$,
where $m_i(\mu)=|\{j|\mu_j=i\}|$.
The generating function of disconnected three partition Hodge integrals is then defined by
\begin{align*}
	G^\bullet(\lambda,\bm w; \mathbf{p}^1, \mathbf{p}^2, \mathbf{p}^3)
	=\exp\Big(\sum_{g\geq0}
	\sum_{\mu^1,\mu^2,\mu^3}
	\lambda^{2g-2+l}
	G_{g,\mu^1,\mu^2,\mu^3}(\bm w)
	\prod_{i=1}^3 p^i_{\mu^i}\Big),
\end{align*}
where $(\mathbf{p}^1, \mathbf{p}^2, \mathbf{p}^3)$ are three families of time variables
and $p^i_{\mu^i}=\prod_{j=1}^{l(\mu^i)}p^i_{\mu^i_j}$.

The above generating function $G^\bullet(\lambda,\bm w; \mathbf{p}^1, \mathbf{p}^2, \mathbf{p}^3)$
of disconnected three partition Hodge integrals admits a Schur function expansion,
whose coefficients are essentially the topological vertex
(see Theorem 8.2 and Conjecture 8.3 in \cite{lllz} and \cite{moop, nt} for proofs).
\begin{align*}
	G^\bullet(\lambda,\bm w; \mathbf{p}^1, \mathbf{p}^2, \mathbf{p}^3)
	=\sum_{\mu^1,\mu^2,\mu^3}
	\tilde{\cW}_{\mu^1,\mu^2,\mu^3}(q)
	\prod_{i=1}^3 q^{
		\frac{1}{2}\kappa_{\mu^i}\frac{w_{i+1}}{w_i}}
	s_{\mu^i}(\mathbf{p}^i),
\end{align*}
where $q=e^{\sqrt{-1}\lambda}$.
The $\tilde{\cW}_{\mu^1,\mu^2,\mu^3}(q)$ is the  topological vertex in the mathematical literature.
The notion of the topological vertex first appear in the string theory literature
It was introduced by Aganagic, Klemm, Mari\~{n}o, and Vafa \cite{akmv} for computing the Gromov-Witten invariants of toric Calabi-Yau threefolds.
Beside the mathematical theory of the topological vertex developed in \cite{lllz},
there appeared also the DT vertex developed in \cite{moop}.
The topological vertex $\tilde{\cW}_{\mu^1,\mu^2,\mu^3}(q)$ used in \cite{lllz} looks slightly different from the original version proposed by physicists \cite{akmv}.
The equivalence of these two expressions of topological vertex was first proved in \cite{moop}  using GW/DT correspondence,
and a combinatorial proof was given in \cite{nt}.

For an arbitrary toric Calabi-Yau threefold $X$,
its toric diagram is a trivalent planar graph which can be constructed by gluing some trivalent vertices together,
and $X$ can be constructed accordingly by gluing some $\bC^3$-pieces together.
The Gromov-Witten invariants of $X$ can be obtained by certain gluing procedure of topological vertices.
In such procedure,
one needs to multiply the topological vertex by some extra factors
to get the framed topological vertex.

\subsection{Topological vertex and framed topological vertex}

Now we recall the definition of the topological vertex given in \cite{akmv}.
Let $\mu^1$, $\mu^2$, $\mu^3$ be three partitions,
then the topological vertex $W_{\mu^1,\mu^2,\mu^3} (q)$ is defined by:
\begin{equation*}
W_{\mu^1,\mu^2,\mu^3}(q) = q^{\kappa_{\mu^2}/2 + \kappa_{\mu^3}/2}
\sum_{\rho^1,\rho^3} c_{\rho^1(\rho^3)^t}^{\mu^1(\mu^3)^t}
\frac{W_{(\mu^2)^t \rho^1}(q) W_{\mu^2 (\rho^3)^t}(q)}{W_{\mu^2}(q)},
\end{equation*}
where:
\begin{equation*}
\begin{split}
& c_{\rho^1(\rho^3)^t}^{\mu^1(\mu^3)^t} = \sum_\eta
c_{\eta \rho^1}^{\mu^1} c_{\eta (\rho^3)^t}^{(\mu^3)^t},\\
& W_{\mu}(q) =q^{\kappa_\mu /4}
\prod_{1\leq i<j \leq l(\mu)} \frac{[\mu_i-\mu_j+j-i]}{[j-i]}
\prod_{i=1}^{l(\mu)} \prod_{v=1}^{\mu_i} \frac{1}{[v-i+l(\mu)]},\\
& W_{\mu,\nu} (q) = q^{|\nu|/2} W_\mu (q) \cdot s_\nu (\cE_\mu(q,t)),
\end{split}
\end{equation*}
and $s_\nu (\cE_\mu(q,t))$ is given by (see \cite[\S 5.2]{zhou5}):
\begin{equation*}
s_\nu (\cE_\mu(q,t)) = s_\nu (q^{\mu_1 -1},q^{\mu_2-2},\cdots).
\end{equation*}
The topological vertex can be rewritten as (see e.g. \cite[Prop. 4.4]{zhou4}):
\be
\label{eqn:def W as skew}
W_{\mu^1,\mu^2,\mu^3}(q)
= (-1)^{|\mu^2|} q^{\kappa_{\mu^3}/2} s_{(\mu^2)^t} (q^{-\rho})
\sum_\eta s_{\mu^1 /\eta} (q^{(\mu^2)^t + \rho}) s_{(\mu^3)^t /\eta} (q^{\mu^2 + \rho}),
\ee
where $\rho = (-\half, -\frac{3}{2}, -\frac{5}{2}, \cdots)$,
and $q^{\mu+\rho}$ denotes the sequence
\begin{equation*}
 q^{\mu+\rho} =
( q^{\mu_1 - \half}, q^{\mu_2-\frac{3}{2}}, \cdots, q^{\mu_l - l+\half},
q^{ -l-\half}, q^{-l-\frac{3}{2}},\cdots )
\end{equation*}
for a partition $\mu = (\mu_1,\mu_2,\cdots,\mu_l)$.

In general
one can also consider the following framed topological vertex with framing $(a_1,a_2,a_3) \in \bZ^3$
introduced in \cite{akmv}:
\be
W_{\mu^1,\mu^2,\mu^3}^{(a_1,a_2,a_3)}(q)
= q^{a_1\kappa_{\mu^1}/2 + a_2\kappa_{\mu^2}/2 + a_3\kappa_{\mu^3}/2} \cdot
W_{\mu^1,\mu^2,\mu^3}(q),
\ee
or more explicitly,
\be
\label{eq-framedTV-def}
\begin{split}
W_{\mu^1,\mu^2,\mu^3}^{(a_1,a_2,a_3)}(q)
=& (-1)^{|\mu^2|} \cdot
q^{a_1\kappa_{\mu^1}/2 + a_2\kappa_{\mu^2}/2 + (a_3+1)\kappa_{\mu^3}/2} \\
& \cdot s_{(\mu^2)^t} (q^{-\rho})
\sum_\eta s_{\mu^1 /\eta} (q^{(\mu^2)^t + \rho}) s_{(\mu^3)^t /\eta} (q^{\mu^2 + \rho}).
\end{split}
\ee
The framed topological vertex encodes the local relative Gromov-Witten invariants of $\bC^3$ with torus action specified by the framing $(a_1,a_2,a_3)$.
Then the original topological vertex is recovered by taking the framing $(0,0,0)$:
\begin{equation*}
W_{\mu^1,\mu^2,\mu^3}(q) = W_{\mu^1,\mu^2,\mu^3}^{(0,0,0)}(q).
\end{equation*}

\subsection{ADKMV Conjecture and framed ADKMV Conjecture}
\label{sec:adkmv}
Since it involves complicated summations over partitions
and specializations of skew Schur functions,
the above expressions of the topological vertex looks very complicated.
It makes the computations of the Gromov-Witten invariants of toric Calabi-Yau threefolds using the topological vertex very involved.
Recall that Schur functions give a basis of the bosonic Fock space,
and thus it is natural to regard the topological vertex as a vector in the $3$-component bosonic Fock space.
In \cite{akmv, adkmv},
Aganagic, Dijkgraaf, Klemm, Mari\~{n}o, and Vafa conjectured that after the boson-fermion correspondence on each of these copies,
  the topological vertex corresponds to an element in the $3$-component fermionic Fock space which can be represented as a Bogoliubov transform
whose coefficients can be written down explicitly:
\begin{Conjecture}
[ADKMV Conjecture \cite{adkmv, akmv}]
Denote:
\begin{equation*}
|\mu^1,\mu^2,\mu^3\rangle = |\mu^1\rangle \otimes |\mu^2\rangle \otimes |\mu^3\rangle
\in \cF^{(0)}\otimes\cF^{(0)}\otimes\cF^{(0)},
\end{equation*}
then the topological vertex is given by:
\be
\begin{split}
W_{\mu^1,\mu^2,\mu^3}(q) =
 \langle \mu^1,\mu^2,\mu^3|
\exp\Big( \sum_{i,j =1,2,3} \sum_{m,n\geq 0} A_{mn}^{ij}(q)
\psi_{-m-\half}^{i} \psi_{-n-\half}^{j*} \Big)
|0\rangle \otimes |0\rangle \otimes |0\rangle,
\end{split}
\ee
where $A^{ij}_{mn}(q)$ are given by equation \eqref{eqn:Aijmn}.
\end{Conjecture}

This formula is called the ADKMV Conjecture.
In \cite{dz1},
Deng and the third-named author have proposed the following framed ADKMV Conjecture which generalizes
the ADKMV Conjecture to the case of framed topological vertex:
\begin{Conjecture}
[framed ADKMV Conjecture \cite{dz1}]
Let $\bm a = (a_1,a_2,a_3)\in \bZ^3$ and let $\mu^1,\mu^2,\mu^3$ be three partitions.
Then the framed topological vertex $W_{\mu^1,\mu^2,\mu^3}^{(\bm a)}(q)$ is given by:
\be
\label{eq-framedADKMV}
\begin{split}
 W_{\mu^1,\mu^2,\mu^3}^{(\bm a)}(q)
= \langle \mu^1,\mu^2,\mu^3|
\exp\Big( \sum_{i,j =1,2,3} \sum_{m,n\geq 0} A_{mn}^{ij}(q; \bm a)
\psi_{-m-\half}^{i} \psi_{-n-\half}^{j*} \Big)
|0\rangle \otimes |0\rangle \otimes |0\rangle,
\end{split}
\ee
where $A_{mn}^{ij} (q;\bm a) = q^{\frac{a_i m(m+1) - a_j n(n+1)}{2}}
\cdot A_{mn}^{ij} (q)$.
\end{Conjecture}

The original ADKMV Conjecture is the special case $(a_1,a_2,a_3)=(0,0,0)$ of the framed ADKMV Conjecture.
The one-legged case (i.e., $\mu^2 = \mu^3 =(0)$ and $a_2=a_3=0$)
and the two-legged case (i.e., $\mu^3 = (0)$ and $a_3=0$)
of the framed ADKMV Conjecture has been proved in \cite{dz1}.
As corollaries,
the generating functions of the open Gromov-Witten invariants of $\bC^3$ with one brane and two branes
are tau-functions of the KP and 2d Toda hierarchies respectively (\cite{djm,ut}).
As a corollary of our proof of the three-legged case of the framed ADKMV Conjecture,
the generating function of the open Gromov-Witten invariants of $\bC^3$ with three branes
is a tau-function of the $3$-component KP hierarchy.

As mentioned in the Introduction,
the topological vertex is a building block of the open Gromov-Witten theory of
a general smooth toric Calabi-Yau threefold $X$,
and many properties of the open Gromov-Witten theory of $X$
can be implied by the properties of the topological vertex by applying
the gluing rules \cite{akmv, adkmv, lllz, dz2}.
It was first proposed by Aganagic-Dijkgraaf-Klemm-Mari\~no-Vafa \cite{adkmv} that
the multi-component KP integrability is supposed to be preserved by the
gluing rule of topological vertices,
and see \cite[(5.22)]{adkmv} for their construction of the fermionic propagator of the gluing procedure.
See also the discussions in \cite{kas}.
This fact was proved by Deng and the third named author,
and the precise statement is as following:
\begin{Theorem}[Theorem 5.2 in \cite{dz2}]
	\label{thm:provided adkmv}
	For any smooth toric Calabi-Yau threefold $X$ whose toric diagram has $n$-legs,
	expand the generating function of open Gromov-Witten invariants of $X$ in terms of Schur basis as
	\begin{align*}
		Z^X(\mathbf{t}^1,\cdots,\mathbf{t}^n)
		=\sum_{\mu^1,\cdots,\mu^n\in\mathcal{P}}
		Z_{\mu^1,\cdots,\mu^n}^X
		\cdot s_{\mu^1}(\mathbf{t}^1)\cdots
		s_{\mu^n}(\mathbf{t}^n).
	\end{align*}
	Provided the framed ADKMV Conjecture,
	and if the toric diagram of $X$ has $g$ loops,
	then there is a Bogoliubov transform of the fermionic vacuum $e^A|0\rangle\otimes\cdots\otimes|0\rangle$ in $\mathcal{F}^n$,
	which is a Laurent series in $g$ formal parameters $\Theta_1,\cdots, \Theta_g$,
	such that
	\begin{align*}
		Z_{\mu^1,\cdots,\mu^n}^X
		=\frac{1}{(2\pi i)^g}
		\oint \frac{\langle \mu^1,\cdots,\mu^n|
			e^A
			|0\rangle \otimes \cdots \otimes |0\rangle}
			{\Theta_1\cdots\Theta_g}
			d\Theta_1\cdots d\Theta_g
	\end{align*}
	for any partitions $\mu^1,\cdots,\mu^n\in\mathcal{P}$.
Therefore,
	the generating function $Z^X(\mathbf{t}^1,\cdots,\mathbf{t}^n)$ of all open Gromov-Witten invariants of $X$ is a part of a tau-function of the $n$-component KP hierarchy.
\end{Theorem}

\subsection{Main result}
\label{sec-outline}

Our main result of this work is the following:
\begin{Theorem}
[= Theorem \ref{thm:Main2}]
The framed ADKMV Conjecture holds for arbitrary partitions $\mu^1$, $\mu^2$, $\mu^3$,
and framing $(a_1,a_2,a_3)\in \bZ^3$.
\end{Theorem}

The rest of this paper will be devoted to proving this result.
For the convenience of the reader,
below we first describe the main steps of the proof,
and then give the details in the next two sections.

\subsubsection{Step 1. An alternative expression of the framed topological vertex as transfer matrix elements}
We first propose the following alternative combinatorial expression for the framed topological vertex as a transfer matrix on the fermionic Fock space.
For $(a_1,a_2,a_3)\in \bZ^3$ and three partitions $\mu^1$, $\mu^2$, $\mu^3$,
the elements of this transfer matrix are defined by:
\begin{equation}
	\label{eq-def-C-000}
C_{\mu^1,\mu^2,\mu^3}^{(a_1,a_2,a_3)} (q) =
\langle \mu^1 | q^{a_1 K/2} \Gamma_{\mu^2}^{(a_2)}(q) q^{-(a_3 +1)K /2} |(\mu^3)^t \rangle,
\end{equation}
where $\Gamma_\mu^{(a)}$ is a product of certain modifications
of the vertex operators $\Gamma_\pm$.
See Definition \ref{def-fermvert} for more details.
Our goal is to prove that this expression is actually equivalent to
the original combinatorial expression of $ W_{\mu^1,\mu^2,\mu^3}^{(a_1,a_2,a_3)} (q)$.

Such a definition is inspired by the following two results in literatures.
First it was noticed in \cite{zhou6} that
\begin{equation} \label{eqn:Two-Partition}
W_{\mu,\nu}(q) = \langle   \mu | q^K \Gamma_+(-q^{-\rho}) \Gamma_-(-q^{-\rho}) q^K |\nu\rangle,
\end{equation}
which essentially leads to $2$d Toda lattice hierarchy.
The definition \eqref{eq-def-C-000} is a generalization of this formula.
On the other hand, Okounkov, Reshetikhin and Vafa  \cite{orv} proved that the topological vertex is equal to the product of the MacMahon function and the partition function of 3D partitions,
see \cite[(3.13), (3.21)]{orv}.
Moreover, the partition function of 3D partitions has a formula in terms of transfer matrix method,
see \cite[(3.17)]{orv}.

\subsubsection{Step 2. Determinantal expression for $ C_{\mu^1,\mu^2,\mu^3}^{(a_1,a_2,a_3)} (q)$}
One of the advantages of the expression for $ C_{\mu^1,\mu^2,\mu^3}^{(a_1,a_2,a_3)} (q)$ is that
the operator $\Gamma_{\mu^2}^{(a_2)}$ is a product of an even number of fermionic operators,
so that we can use Wick's Theorem to  show that $ C_{\mu^1,\mu^2,\mu^3}^{(a_1,a_2,a_3)} (q)$ can be represented as a determinant
(see \S \ref{sec-det-C}-\S \ref{eq-computeF} for details):
\begin{equation*}
C_{\mu^1,\mu^2,\mu^3}^{(a_1,a_2,a_3)} (q) =
\det \big( \tF_{kl}(q; \bm a) \big)_{1\leq k,l\leq r_1+r_2+r_3},
\end{equation*}
where $r_i$ is the number of boxes lie in the diagonal of the Young diagram associated with $\mu^i$
for $i=1,2,3$.

\subsubsection{ Step 3. The equivalence between different combinatorial expressions of the framed topological vertex}
The next step is to show the  equivalence between the alternative combinatorial expression and the original combinatorial expression of the framed topological vertex:
\begin{equation*}
C_{\mu^1,\mu^2,\mu^3}^{(a_1,a_2,a_3)} (q) = W_{\mu^1,\mu^2,\mu^3}^{(a_1,a_2,a_3)} (q),
\end{equation*}
for every $(a_1,a_2,a_3)\in \bZ^3$ and partitions $\mu^1$, $\mu^2$, $\mu^3$,
where $W_{\mu^1,\mu^2,\mu^3}^{(a_1,a_2,a_3)} (q)$ is the original framed topological vertex
given by \eqref{eq-framedTV-def}.
The proof of this equivalence will be given in \S \ref{sec-equiv-vertex}.
It also relies on the special form of the operator  $\Gamma_{\mu^2}^{(a_2)}$ and a special case of  the determinantal formula in Step 2.

\subsubsection{ Step 4.  Determinantal formula for the right-hand side of the framed ADKMV Conjecture}

Denote by $B_{\mu^1,\mu^2,\mu^3}^{(\bm a)}(q)$ the right-hand side of
the framed ADKMV Conjecture \eqref{eq-framedADKMV}.
Then we show that this Bogoliubov transform also satisfies a determinantal formula
(see \S \ref{sec-det-Bog} for details):
\begin{equation*}
B_{\mu^1,\mu^2,\mu^3}^{(\bm a)}(q) = \det\big( B_{kl}(q;\bm a) \big)_{1\leq k,l\leq r_1+r_2+r_3}.
\end{equation*}
The entries $B_{kl}$ are given by certain reordering of $(-1)^n \cdot A_{mn}^{ij}$
where $A_{mn}^{ij}$ are the coefficients of the Bogoliubov transform.

 \subsubsection{Step 5. The equivalence of the determinants}

It turns out that
$$\tF_{kl}(q;\bm a) \neq   B_{kl}(q;\bm a).$$
To finish the proof,
we show that the two determinantal formulas actually give the same result
(see \S \ref{sec-match-coeff} for details):
\begin{equation*}
\det \big(\tF_{kl}(q;\bm a) \big)= \det \big(B_{kl}(q;\bm a)\big).
\end{equation*}
And this completes the proof of the ADKMV Conjecture.

\section{An Alternative Combinatorial Expression of Framed Topological Vertex and Its Determinantal Formula}
\label{sec-ferm-topovert}

In this section,
we deal with the left-hand side of the framed ADKMV Conjecture \eqref{eq-framedADKMV} and prove Theorem \ref{thm:Main1}.
More precisely,
we first introduce the alternative combinatorial expression $C_{\mu^1,\mu^2,\mu^3}^{(a_1,a_2,a_3)}(q)$
of the framed topological vertex,
and then derive a determinantal formula for it.
Then we show that this  expression is actually equivalent to
the original expression for the framed topological vertex $W_{\mu^1,\mu^2,\mu^3}^{(a_1,a_2,a_3)}(q)$ defined by \eqref{eq-framedTV-def},
and thus we also obtain a determinantal formula for $W_{\mu^1,\mu^2,\mu^3}^{(a_1,a_2,a_3)}(q)$.

\subsection{A boson-fermionic field assignment}
\label{sec:bf assign}
In this subsection we introduce a boson-fermionic field assignment,
which assigns an operator $\Psi_\mu^{(a)}(q)$ to each partition $\mu$ and framing $a\in\bZ$.
This operator will be used in the next subsection
to construct an alternative combinatorial expression
of the framed topological vertex.

Recall the bosonic Fock space has a basis $\{s_\mu\}$  indexed by all partitions of integers.
The boson-fermion correspondence assigns to each bosonic state
 a fermionic state, obtained by the action of a fermionic operator on the fermionic vacuum:
 \ben
s_\mu \mapsto  (-1)^{n_1+\cdots+n_k}\cdot
\psi_{-m_1-\half} \psi_{-n_1-\half}^* \cdots
\psi_{-m_k-\half} \psi_{-n_k-\half}^* |0\rangle,
\een
where $m_1, \cdots, m_k$ and $n_1, \cdots, n_k$ are the Frobenius coordinates of $\mu$:
$$\mu=(m_1,\cdots, m_k | n_1,\cdots,n_k).$$
This correspondence can be pictorially visualized using the Dirac sea \cite{djm}.
In this picture,
the partition $\mu$ corresponds to a system which has electrons at energy levels $n_1+\frac{1}{2}, \cdots, n_k + \frac{1}{2}$,
and positrons at energy levels  $-(m_1+\frac{1}{2}), \cdots, -(m_k+\frac{1}{2})$.
The fermionic states are actually given by semi-infinite wedge product,
and in particular,
 $|0\rangle$ is given by an infinite wedge:
$$|0\rangle = z^{\half}\wedge z^{\frac{3}{2}}
\wedge z^{\frac{5}{2}} \wedge \cdots \in \cF^{(0)},$$
the effects of the operators  $\psi_{-m_1-\half} , \cdots, \psi_{-m_k-\half}$ and $\psi_{-n_1-\half}^* ,\cdots, \psi_{-n_k-\half}^*$ are essentially
the modification of the vacuum at the corresponding energy levels.

We now  introduce a boson-fermionic field assignment which assigns a fermionic field
to each bosonic state by similar ideas as follows.
For simplicity,
from now on we denote by $\Gamma_\pm(z)$ the operators $\Gamma_\pm(\{z\})$
(where $\{z\}=(z,\frac{z^2}{2},\frac{z^3}{3},\cdots)$) for a single formal variable $z$.
We first define the field assigned to the state indexed by the empty partition as follows:
\be \label{eqn:Gamma-Emptyset}
\begin{split}
\Psi_{\emptyset}  (q) =\prod_{j\geq 0}  \Gamma_{-}(q^{-j-\half})
\cdot \prod_{i\geq 0}  \Gamma_{+ }(q^{-i-\half})
= \Gamma_-(\bm t (q^\rho))  \Gamma_+(\bm t ( q^\rho)),
\end{split}
\ee
which is a well-defined operator on the fermionic Fock space since
$$\bm t (q^\rho)=(t_1(q^\rho),t_2(q^\rho),t_3(q^\rho),\cdots)$$ and $t_n(q^\rho)=\frac{1}{n}p_n(q^\rho)$ are rational functions of $q$ given by equation \eqref{eqn:p_n at q^rho}.

For a partition $\mu = (m_1,\cdots,m_k|n_1,\cdots,n_k)$ and an integer $a\in \bZ$,
we define $\Psi_\mu^{(a)}(q)$ by modifying the field $\Psi_\emptyset(q)$ in the following way.
For $j=1,\cdots, k$,
we change the factor $\Gamma_-(q^{-(m_j+\frac{1}{2})})$ to $\widetilde{\Gamma}_+^{(a)}(q^{m_j+\frac{1}{2}})$ to be defined below,
and we change the factor $\Gamma_+(q^{-(n_j+\frac{1}{2})})$ to $\widetilde{\Gamma}_-^{(a)}(q^{n_j+\frac{1}{2}})$ to be defined below.
We define the modified vertex operators $\tGa_\pm^{(a)}(z)$ to be:
\be
\label{eq-def-tGamma}
\begin{split}
&\tGa_+^{(a)} (z) = f_1^{(a)}(z) \Gamma_+(z) R^{-1} z^C,\\
&\tGa_-^{(a)} (z) = f_2^{(a)}(z) z^C R \Gamma_-(z),
\end{split}
\ee
where $\Gamma_\pm$ are the vertex operators \eqref{def-vert-Gamma},
and the functions $f_1,f_2$ are given by:
\be
\label{eq-def-f1f2}
\begin{split}
&f_1(q^k) = (-1)^{k-\half} \cdot q^{\frac{(a+1)k^2}{2} - \frac{k}{2} -\frac{1}{16}},\\
&f_2(q^k) = (-1)^{k+\half} \cdot q^{-\frac{a k^2}{2} - \frac{k}{2} -\frac{1}{16}},
\end{split}
\ee
for arbitrary $k$
(and in what follows we will always consider $z$ of the form $z=q^k$ where $k\in \bZ+\half$).
Using the relations \eqref{eq-psi-inGamma} and the identity
\begin{equation*}
z^{-C} R \psi^* (z^{-1}) = \psi^*(z^{-1}) z^{-C} R,
\end{equation*}
we may rewrite the modified vertex operators $\tGa_\pm^{(a)}(z)$ as:
\be
\label{eq-mod-vert}
\begin{split}
&\tGa_+^{(a)} (z) = f_1^{(a)}(z) \Gamma_-(z^{-1}) \psi^*(z^{-1}),\\
&\tGa_-^{(a)} (z) = z\cdot f_2^{(a)}(z) \psi(z) \Gamma_+(z^{-1}).
\end{split}
\ee

To summarize, we make the following definition for the fermionic field $\Psi_{\mu}^{(a)}(q)$
assigned to a partition $\mu$ and a framing parameter $a\in \bZ$.

\begin{Definition}
\label{def-boson-ferm}

For  a partition $\mu = (m_1,\cdots,m_k|n_1,\cdots,n_k)$ and an integer $a\in \bZ$,
the operator $\Psi_\mu^{(a)} (q)$ is defined by:
\begin{equation*}
\begin{split}
\Psi_{\mu}^{(a)} (q) =& \prod_{j\geq 0}^{\longleftarrow} \Gamma_{-,\{j,\mu\}}^{(a)}(q^{-j-\half})
\cdot \prod_{i\geq 0}^{\longrightarrow} \Gamma_{+,\{i,\mu\}}^{(a)}(q^{-i-\half})\\
=&\ \cdots \Gamma_{-,\{2,\mu\}}^{(a)}(q^{-\frac{5}{2}}) \Gamma_{-,\{1,\mu\}}^{(a)}(q^{-\frac{3}{2}})
\Gamma_{-,\{0,\mu\}}^{(a)}(q^{-\frac{1}{2}})
\Gamma_{+,\{0,\mu\}}^{(a)}(q^{-\frac{1}{2}})
\Gamma_{+,\{1,\mu\}}^{(a)}(q^{-\frac{3}{2}}) \cdots,
\end{split}
\end{equation*}
where
\be
\begin{split}
\Gamma_{-,\{j,\mu\}}^{(a)}(z) =
\begin{cases}
\Gamma_-(z), & \text{ if $j\notin \{m_1,\cdots,m_k\}$};\\
\tGa_+^{(a)}(z^{-1}), & \text{ if $j\in \{m_1,\cdots,m_k\}$},
\end{cases}
\end{split}
\ee
and
\be
\begin{split}
\Gamma_{+,\{i,\mu\}}^{(a)}(z) =
\begin{cases}
\Gamma_+(z), & \text{ if $i\notin \{n_1,\cdots,n_k\}$};\\
\tGa_-^{(a)}(z^{-1}), & \text{ if $i\in \{n_1,\cdots,n_k\}$.}
\end{cases}
\end{split}
\ee
\end{Definition}

\subsection{A new combinatorial expression for the framed topological vertex}
\label{sec:C as Psi}
We now use the boson-fermionic field assignment introduced in last subsection to
propose a new combinatorial expression for the framed topological vertex as follows.

\begin{Definition}
\label{def-fermvert}
Given $(a_1,a_2,a_3)\in \bZ^3$ and three partitions $\mu^1,\mu^2,\mu^3$,
we define  $C_{\mu^1,\mu^2,\mu^3}^{(a_1,a_2,a_3)}(q)$ to be
the following combinatorial expression:
\be
\label{eq-def-fermvert}
C_{\mu^1,\mu^2,\mu^3}^{(a_1,a_2,a_3)} (q) =
\langle \mu^1 | q^{a_1 K/2} \Psi_{\mu^2}^{(a_2)}(q) q^{-(a_3 +1)K /2} |(\mu^3)^t \rangle,
\ee
where $K$ is the cut-and-join operator \eqref{eq-def-C&Jopr}.
\end{Definition}

\begin{Remark}
Note when $\mu^2=\emptyset$,
\eqref{eq-def-fermvert} reduces to \eqref{eqn:Two-Partition}.

The idea of formulating the topological vertex using
the vertex operators $\Gamma_\pm (\bm t)$
has been investigated by Okounkov, Reshetikhin, and Vafa,
see \cite[(3.17), (3.21)]{orv}.
See also \cite{wy,ikv,fw} and references therein for some
other combinatorial expressions of the topological vertex.
However, all these constructions are slightly different from the original topological vertex in \cite{akmv, adkmv},
and thus it seems not easy to relate these combinatorial expressions directly
to the Bogoliubov transform conjectured in \cite{adkmv}.

Our proposed expression for the framed topological vertex \eqref{eq-def-fermvert}
is inspired by these works \cite{orv, wy, ikv, fw}
but different from their constructions.
We will see later that our expression
equals exactly to the original framed topological vertex
and possesses a nice determinantal formula.
\end{Remark}

\subsection{Determinantal formula for $C_{\mu^1,\mu^2,\mu^3}^{(a_1,a_2,a_3)}(q)$}
\label{sec-det-C}

In this subsection we use Wick's Theorem to derive a determinantal formula for $C_{\mu^1,\mu^2,\mu^3}^{(a_1,a_2,a_3)}(q)$.

First for $m,n\geq 0$, we denote:
\be
\begin{split}
& \Psi_{m,\emptyset}^{(a)}(q) =
\prod_{j=m+1}^\infty \Gamma_- (q^{-j-\half}) \cdot \tGa_+^{(a)} (q^{m+\half}) \cdot
\prod_{j=0}^{m-1}\Gamma_-(q^{-j-\half}) \cdot\prod_{j=0}^\infty \Gamma_+(q^{-j-\half}),\\
& \Psi_{\emptyset,n}^{(a)}(q) = \prod_{j=0}^\infty \Gamma_- (q^{-j-\half}) \cdot
\prod_{j=0}^{n-1}\Gamma_+ (q^{-j-\half}) \cdot \tGa_-^{(a)} (q^{n+\half}) \cdot
\prod_{j=n+1}^{\infty}\Gamma_+(q^{-j-\half}).
\end{split}
\ee
And denote:
\be
\label{eq-def-entryF}
F_{mn}^{ij}(q;\bm a) = \begin{cases}
\langle (m|n) | q^{a_1 K/2} \Psi_\emptyset^{(a_2)} (q) |0\rangle, &\text{$(i,j)=(1,1)$;}\\
\langle  \Psi_{(m|n)}^{(a_2)} (q) \rangle, &\text{$(i,j)=(2,2)$;}\\
\langle 0 | \Psi_{\emptyset}^{(a_2)} (q) q^{-(a_3+1)K/2} |(n|m)\rangle, &\text{$(i,j)=(3,3)$;}\\
-\langle  \psi^*_{m+\half} q^{a_1 K /2}\Psi_{\emptyset,n}^{(a_2)} (q) \rangle, &\text{$(i,j)=(1,2)$;}\\
-\langle  \psi^*_{m+\half} q^{a_1 K/2}\Psi_{\emptyset}^{(a_2)} (q)
 q^{-(a_3+1)K/2} \psi_{-n-\half} \rangle, &\text{$(i,j)=(1,3)$;}\\
 -(-1)^{n}\langle  \psi_{n+\half} q^{a_1 K /2}
 \Psi_{m,\emptyset}^{(a_2)} (q) \rangle, &\text{$(i,j)=(2,1)$;}\\
\langle  \Psi_{m,\emptyset}^{(a_2)} (q)  q^{-(a_3+1)K/2} \psi_{-n-\half} \rangle, &\text{$(i,j)=(2,3)$;}\\
(-1)^{m+n}\langle  \psi_{n+\half}  q^{a_1 K/2} \Psi_{\emptyset}^{(a_2)} (q)
 q^{-(a_3+1)K/2} \psi_{-m-\half}^* \rangle,
&\text{$(i,j)=(3,1)$;}\\
(-1)^{m}\langle \Psi_{\emptyset,n}^{(a_2)} (q)
q^{-(a_3+1)K/2} \psi_{-m-\half}^* \rangle, &\text{$(i,j)=(3,2)$,}
\end{cases}
\ee
for $m,n\geq 0$ and $i,j=1,2,3$.

Now fix three partitions
$\mu^1,\mu^2,\mu^3$ whose Frobenius notations are:
\begin{equation*}
\mu^i = ( m_1^i, \cdots, m_{r_i}^i | n_1^i, \cdots, n_{r_i}^i ),
\qquad i=1,2,3.
\end{equation*}
Denote by $I_1,I_2,I_3$ the following sets of indices:
\be
\begin{split}
&I_1 = \{1,2,\cdots,r_1\},\\
&I_2 = \{r_1 +1, r_1+2,\cdots,r_1+r_2\},\\
&I_3 = \{r_1 +r_2+1, r_1+r_2+2,\cdots,r_1+r_2+r_3\},
\end{split}
\ee
For an arbitrary integer $k$ with $1\leq k\leq r_1+r_2+r_3$,
we denote:
\be
\bar k = k - \sum_{e=1}^{i-1} r_e \in \{1,2,\cdots,r_i\},
\qquad \text{if $k\in I_i$}.
\ee
Then we have the following
(see \S \ref{sec-det-C} for the proof):
\begin{Theorem}
\label{prop-det-C}
For arbitrary $(a_1,a_2,a_3)\in \bZ^3$ and partitions $\mu^1,\mu^2,\mu^3$,
the combinatorial expression
$C_{\mu^1,\mu^2,\mu^3}^{(a_1,a_2,a_3)} (q)$ is given by
the following determinantal formula:
\be
\label{eq-det-C}
C_{\mu^1,\mu^2,\mu^3}^{(a_1,a_2,a_3)} (q) =
\det \big( \tF_{kl}(q; \bm a) \big)_{1\leq k,l\leq r_1+r_2+r_3},
\ee
where the entries $\{\tF_{k,l}(q;\bm a)\}_{1\leq k,l\leq r_1+r_2+r_3}$ are given by:
\be
\tF_{kl}(q;\bm a) = F_{m_{\bar k}^i n_{\bar l}^j}^{ij}(q;\bm a),
\qquad \text{if $k\in I_i$ and $l\in I_j$.}
\ee
\end{Theorem}
\begin{proof}
Let $a\in \bZ$ and let $\mu=(m_1,\cdots,m_r|n_1,\cdots,n_r)$ be a partition.
Using the identities \eqref{eq-mod-vert},
we can rewrite the operator $\Psi_{\mu}^{(a)} (q)$ in Definition \ref{def-fermvert} as:
\begin{equation*}
\begin{split}
& c_\mu^{(a)}(q)\cdot
 \prod_{j=m_1+1}^\infty \Gamma_-(q^{-j-\half}) \cdot
\Big( \Gamma_-(q^{-m_1-\half}) \psi^*(q^{-m_1-\half}) \Big) \cdot
\prod_{j=m_2+1}^{m_1-1} \Gamma_-(q^{-j-\half})\cdot \\
& \Big( \Gamma_-(q^{-m_2-\half}) \psi^*(q^{-m_2-\half}) \Big) \cdots
\Big( \Gamma_-(q^{-m_r-\half}) \psi^*(q^{-m_r-\half}) \Big) \cdot
\prod_{j=0}^{m_r-1} \Gamma_-(q^{-j-\half}) \cdot\\
&\prod_{j=0}^{n_r-1} \Gamma_+(q^{-j-\half}) \cdot \Big( \psi^*(q^{n_r+\half})\Gamma_+(q^{-n_r-\half})  \Big)
\cdot \prod_{j=n_r+1}^{n_{r-1}-1} \Gamma_+(q^{-j-\half}) \cdots \\
& \prod_{j= n_2+1}^{n_1-1}\Gamma_+(q^{-j-\half}) \cdot
\Big( \psi^*(q^{n_1+\half})\Gamma_+(q^{-n_1-\half})  \Big) \cdot
\prod_{j= n_1+1}^{\infty}\Gamma_+(q^{-j-\half}),
\end{split}
\end{equation*}
where the coefficient $c_\mu^{(a)}(q)$ is given by:
\be
c_\mu^{(a)}(q) = \prod_{k=1}^r f_1^{(a)} (q^{m_k+\half}) \cdot
\prod_{k=1}^r \big( q^{n_k+\half} \cdot f_2^{(a)} (q^{n_k+\half}) \big).
\ee

Now by the property \eqref{eq-eigen-C&J} we know that:
\begin{equation*}
\begin{split}
& \langle \mu^1 | q^{a_1 K/2} = q^{a_1\kappa_{\mu^1} /2}\cdot \langle \mu^1 |,\\
& q^{-(a_3 +1)K /2} |(\mu^3)^t \rangle = q^{-(a_3 +1)\kappa_{(\mu^3)^t} /2} \cdot |(\mu^3)^t \rangle
= q^{(a_3 +1)\kappa_{\mu^3} /2} \cdot |(\mu^3)^t \rangle,
\end{split}
\end{equation*}
and thus $C_{\mu^1,\mu^2,\mu^3}^{(a_1,a_2,a_3)} (q)$ can be rewritten as:
\begin{equation*}
\begin{split}
C_{\mu^1,\mu^2,\mu^3}^{(a_1,a_2,a_3)} (q) =&
q^{a_1\kappa_{\mu^1} /2 + (a_3 +1)\kappa_{\mu^3} /2}\cdot
\langle \mu^1 |  \Psi_{\mu^2}^{(a_2)}(q)  |(\mu^3)^t \rangle \\
=& q^{a_1\kappa_{\mu^1} /2 + (a_3 +1)\kappa_{\mu^3} /2}\cdot c_{\mu^2}^{(a_2)}(q) \\
&  \cdot \Big\langle \mu^1 \Big| \prod_{j=m_1^2+1}^\infty \Gamma_-(q^{-j-\half}) \cdot
\Big( \Gamma_-(q^{-m_1^2-\half}) \psi^*(q^{-m_1^2-\half}) \Big) \cdots\\
& \qquad \cdots  \Big( \psi(q^{n_1^2+\half})\Gamma_+(q^{-n_1^2-\half})  \Big) \cdot
\prod_{j= n_1^2+1}^{\infty}\Gamma_+(q^{-j-\half})
\Big| (\mu^3)^t \Big\rangle.
\end{split}
\end{equation*}

Assume $\mu^i = ( m_1^i, \cdots, m_{r_i}^i | n_1^i, \cdots, n_{r_i}^i )$ for $i=1,2,3$.
By \eqref{eq-ferm-basismu} we have:
\begin{equation*}
\begin{split}
&\langle \mu^1 | = (-1)^{\sum_{k=1}^{r_1} n_k^1} \cdot \langle 0|
\psi_{n_{r_1}^1+\half} \psi^*_{m_{r_1}^1+\half} \cdots
\psi_{n_{1}^1+\half} \psi^*_{m_{1}^1+\half},\\
&| (\mu^3)^t \rangle = (-1)^{\sum_{k=1}^{r_3} m_k^3} \cdot
\psi_{-n_{1}^3-\half} \psi^*_{-m_{1}^3-\half} \cdots
\psi_{-n_{r_3}^3-\half} \psi^*_{-m_{r_3}^3-\half} |0\rangle,\\
\end{split}
\end{equation*}
and then:
\be
\label{eq-fermvert-2}
\begin{split}
&C_{\mu^1,\mu^2,\mu^3}^{(a_1,a_2,a_3)} (q)
= (-1)^{\sum_{k=1}^{r_1} n_k^1 + \sum_{k=1}^{r_3} m_k^3} \cdot
 q^{a_1\kappa_{\mu^1} /2 + (a_3 +1)\kappa_{\mu^3} /2}\cdot c_{\mu^2}^{(a_2)}(q)  \\
&\quad  \cdot \Big\langle  \prod_{k=1}^{r_1} \psi_{n_{k}^1+\half} \psi^*_{m_{k}^1+\half} \cdot
\prod_{j=m_1^2+1}^\infty \Gamma_-(q^{-j-\half}) \cdot
\Big( \Gamma_-(q^{-m_1^2-\half}) \psi^*(q^{-m_1^2-\half}) \Big) \cdots\\
&\quad\qquad \cdots  \Big( \psi(q^{n_1^2+\half})\Gamma_+(q^{-n_1^2-\half})  \Big)  \cdot
\prod_{j= n_1^2+1}^{\infty}\Gamma_+(q^{-j-\half}) \cdot
\prod_{k=1}^{r_3} \psi_{-n_{k}^3-\half} \psi^*_{-m_{k}^3-\half}
 \Big\rangle.
\end{split}
\ee

Now we compute the vacuum expectation value in the right-hand side.
Notice that:
\be
\langle 0| \Gamma_-(z) = \langle 0|,
\qquad
\Gamma_+(z) | 0\rangle = |0\rangle.
\ee
Moreover,
by the relation \eqref{eq-comm-alphapsi} and the Baker-Campbell-Hausdorff formula \eqref{eq-BCH},
one may see that the conjugation
\begin{equation*}
\Gamma_\pm (z)^{-1} \psi_s \Gamma_\pm (z),
\qquad
\Gamma_\pm (z)^{-1} \psi_s^* \Gamma_\pm (z),
\end{equation*}
of the free fermions $\psi_s$, $\psi_s^*$ by the operators $\Gamma_\pm (z)$
are all linear (infinite) summations of $\psi_s$ and $\psi_s^*$,
and thus one can apply Wick's Theorem \eqref{eq-Wickthm} to compute the right-hand side of \eqref{eq-fermvert-2}.
Or more precisely,
we need to consider all possible ways of grouping (the conjugations of) the following free fermions into pairs:
\begin{equation*}
\begin{split}
&\prod_{k=1}^{r_1} \psi_{n_{k}^1+\half} \psi^*_{m_{k}^1+\half},
\quad
\psi^*(q^{-m_1^2-\half}),\quad \cdots \quad, \quad \psi^*(q^{-m_{r_2}^2-\half}),\\
&\psi(q^{n_{r_2}^2 +\half}), \quad \cdots\quad, \quad \psi(q^{n_1^2 +\half}), \quad
\prod_{k=1}^{r_3} \psi_{-n_{k}^3-\half} \psi^*_{-m_{k}^3-\half}.
\end{split}
\end{equation*}
Moreover,
by \eqref{eq-quadVEV} we see that
the pairing of two $\psi$'s or two $\psi^*$'s leads to zero,
and thus we only need to consider the pairing of one $\psi$ with one $\psi^*$.
Each possible pair of free fermions is of one of the following nine forms:
\begin{itemize}
\item[1)]
The paring of $\psi_{n_{k}^1+\half}$ (where $1\leq k \leq r_1$) with one $\psi_s^*$
(which will be denoted by $G_{kl}^{11}$, $G_{kl}^{12}$, $G_{kl}^{13}$ respectively):
\begin{equation*}
\begin{split}
G_{kl}^{11} =
& \Big\langle \psi_{n_k^1 +\half } \psi^*_{m_l^1 +\half } \prod_{j=0}^\infty \Gamma_-(q^{-j-\half}) \Big\rangle,
\quad  1\leq l\leq r_1; \\
G_{kl}^{12} =
& \Big\langle \psi_{n_k^1 +\half } \prod_{j=m_l^2}^\infty \Gamma_-(q^{-j-\half})  \cdot
\psi^* (q^{-m_l^2 -\half}) \prod_{j=0}^{m_l^2-1} \Gamma_-(q^{-j-\half})
\Big\rangle,
\quad  1\leq l\leq r_2; \\
G_{kl}^{13} =
& \Big\langle \psi_{n_k^1 +\half } \prod_{j=0}^\infty \Gamma_-(q^{-j-\half})  \cdot
\prod_{j=0}^\infty \Gamma_+(q^{-j-\half}) \cdot
\psi^*_{-m_l^3 -\half} \Big\rangle,
\quad  1\leq l\leq r_3. \\
\end{split}
\end{equation*}

\item[2)]
The paring of $\psi(q^{n_k^2 +\half})$ (where $1\leq k \leq r_2$) with one $\psi_s^*$
(which will be denoted by $G_{kl}^{21}$, $G_{kl}^{22}$, $G_{kl}^{23}$ respectively):
\begin{equation*}
\begin{split}
&G_{kl}^{21} = -
 \Big\langle \psi^*_{m_l^1 +\half } \prod_{j=0}^\infty \Gamma_-(q^{-j-\half})
\cdot \prod_{j=0}^{n_k^2 -1} \Gamma_+(q^{-j-\half}) \cdot
\psi(q^{n_k^2 +\half})
 \Big\rangle,
\quad  1\leq l\leq r_1; \\
&G_{kl}^{22} = -
 \Big\langle \psi^*( q^{-m_l^2 -\half })
 \prod_{j=0}^{m_l^2-1} \Gamma_-(q^{-j-\half})  \prod_{j=0}^{n_k^2-1} \Gamma_+(q^{-j-\half})
\cdot \psi(q^{n_k^2 +\half})
\Big\rangle,
\quad  1\leq l\leq r_2; \\
&G_{kl}^{23} =
 \Big\langle
\prod_{j=0 }^{n_k^2-1} \Gamma_+(q^{-j-\half}) \cdot
\psi(q^{n_k^2 +\half}) \prod_{j=n_k^2 }^\infty \Gamma_+(q^{-j-\half})
 \psi^*_{-m_l^3 -\half }  \Big\rangle,
\quad  1\leq l\leq r_3. \\
\end{split}
\end{equation*}

\item[3)]
The paring of $\psi_{-n_k^3 -\half}$ (where $1\leq k \leq r_3$) with one $\psi_s^*$
(which will be denoted by $G_{kl}^{21}$, $G_{kl}^{22}$, $G_{kl}^{23}$ respectively):
\begin{equation*}
\begin{split}
&G_{kl}^{31} =-
 \Big\langle \psi^*_{m_l^1 +\half } \prod_{j=0}^\infty \Gamma_-(q^{-j-\half})
\cdot \prod_{j=0}^\infty \Gamma_+(q^{-j-\half}) \cdot
\psi_{-n_k^3 -\half}
 \Big\rangle,
\quad  1\leq l\leq r_1; \\
&G_{kl}^{32} =-
 \Big\langle \psi^*( q^{-m_l^2 -\half })
 \prod_{j=0}^{m_l^2-1} \Gamma_-(q^{-j-\half})  \prod_{j=0}^\infty \Gamma_+(q^{-j-\half})
\cdot \psi_{-n_k^3 -\half}
\Big\rangle,
\quad  1\leq l\leq r_2; \\
&G_{kl}^{33} =
 \Big\langle
 \prod_{j=0 }^\infty \Gamma_+(q^{-j-\half})
\psi_{-n_k^3 -\half} \psi^*_{-m_l^3 -\half }  \Big\rangle,
\quad  1\leq l\leq r_3. \\
\end{split}
\end{equation*}

\end{itemize}

Now denote $r = r_1+r_2+r_3$.
Then by applying Wick's Theorem to the right-hand side of \eqref{eq-fermvert-2} we get
the following determinantal formula for $C_{\mu^1,\mu^2,\mu^3}^{(a_1,a_2,a_3)} (q)$:
\begin{equation*}
\begin{split}
C_{\mu^1,\mu^2,\mu^3}^{(a_1,a_2,a_3)} (q)
=  (-1)^{\sum_{k=1}^{r_1} n_k^1 + \sum_{k=1}^{r_3} m_k^3}
 q^{a_1\kappa_{\mu^1} /2 + (a_3 +1)\kappa_{\mu^3} /2} c_{\mu^2}^{(a_2)}(q)
\cdot (-1)^{r_2}   \det (G),
\end{split}
\end{equation*}
where the factor $(-1)^{r_2}$ is the sign of the permutation
\begin{equation*}
(2,4,\cdots,2r_2-2,2r_2,2r_2-1,2r_2-3,\cdots,3,1)
\end{equation*}
of $(1,2,\cdots, 2r_2)$
(which occurs because of the fact that the order of the $\psi$'s and $\psi^*$'s in
the right-hand side of \eqref{eq-fermvert-2}
differs from the standard order in \eqref{eq-Wick-det-2} by such a permutation),
and $G$ is the following matrix of size $r\times r$:
\begin{equation*}
G=
\left(
\begin{array}{ccc}
G^{11} & G^{12} & G^{13} \\
G^{21} & G^{22} & G^{23} \\
G^{31} & G^{32} & G^{33} \\
\end{array}
\right).
\end{equation*}
The block $G^{ij}$ is just the $r_i\times r_j$ matrix
$(G^{ij}_{kl})_{1\leq k \leq r_i, 1\leq l\leq r_j}$.

Now in order to prove Proposition \ref{prop-det-C},
we only need to show that:
\be
\label{eq-det-toprove-1}
\begin{split}
&\det \big( \tF_{kl}(q; \bm a) \big)_{1\leq k,l\leq r} \\
= &
(-1)^{\sum_{k=1}^{r_1} n_k^1 + \sum_{k=1}^{r_3} m_k^3 +  r_2 } \cdot
 q^{a_1\kappa_{\mu^1} /2 + (a_3 +1)\kappa_{\mu^3} /2  }\cdot c_{\mu^2}^{(a_2)}(q) \cdot \det (G).
\end{split}
\ee
Recall that the $r\times r$ matrix $(\tF_{kl})$ is of the form
\begin{equation*}
\tF=
\left(
\begin{array}{ccc}
F^{11} & F^{12} & F^{13} \\
F^{21} & F^{22} & F^{23} \\
F^{31} & F^{32} & F^{33} \\
\end{array}
\right).
\end{equation*}
where the entries of blocks $F^{ij}$ are given by \eqref{eq-def-entryF}.
We claim that the entries $\tF_{kl}$ differ from the entries of the transpose $G^t$ only by some simple factors
which match with the coefficient in the right-hand side of \eqref{eq-det-toprove-1}.
In fact,
by \eqref{eq-eigen-C&J} we know that:
\begin{equation*}
\begin{split}
F_{kl}^{11} = &
\big\langle (m_k^1|n_l^1) \big| q^{a_1 K/2} \prod_{j=0}^\infty \Gamma_- (q^{-j-\half}) \big|0\big\rangle \\
=& q^{a_1 \kappa_{(m_k^1|n_l^1)}/2} \big\langle (m_k^1|n_l^1) \big| \prod_{j=0}^\infty \Gamma_- (q^{-j-\half}) \big|0\big\rangle\\
=& (-1)^{n_l^1} \cdot q^{\half a_1 (m_k^1(m_k^1 +1)) + n_l^1(n_l^1 +1))} \cdot
\big\langle \psi_{n_l+\half} \psi_{m_k +\half}^* \cdot\prod_{j=0}^\infty \Gamma_- (q^{-j-\half})\big\rangle \\
=& (-1)^{n_l^1} \cdot q^{\half a_1 (m_k^1(m_k^1 +1)) + n_l^1(n_l^1 +1))} \cdot G_{lk}^{11} ,
\end{split}
\end{equation*}
Similarly,
one can check the following relations case by case using \eqref{eq-comm-expKpsi} and \eqref{eq-eigen-C&J}:
\begin{equation*}
\begin{split}
& F_{kl}^{22} = - f_1^{(a_2)} (q^{m_k^2+\half}) \cdot
 q^{n_l^2+\half} \cdot f_2^{(a_2)} (q^{n_l^2+\half}) \cdot G_{lk}^{22},\\
& F_{kl}^{33} = (-1)^{m_k^3}\cdot q^{-\half (a_3+1) (m_k^3(m_k^3 +1)) + n_l^3(n_l^3 +1))} \cdot G_{lk}^{33},\\
& F_{kl}^{12} = q^{\half a_1(m_k^1+\half)^2 }q^{n_l^2+\half} f_2(q^{n_l^2 +\half}) \cdot G_{lk}^{21},\\
& F_{kl}^{21} = -(-1)^{n_l^1 } \cdot q^{- \half a_1 (n_l^1 +\half)^2}
f_1(q^{m_k^2  +\half}) \cdot G_{lk}^{12},\\
& F_{kl}^{13} =  q^{\half a_1(m_k^1 +\half)^2} q^{-\half (a_3+1)(n_l^3 +\half)^2}
G_{lk}^{31},\\
& F_{kl}^{31} = (-1)^{m_k^3 + n_l^1} \cdot
q^{- \half a_1 (n_l^1 +\half)^2}q^{\half (a_3+1)(m_k^3 +\half)^2} \cdot G_{lk}^{13},\\
& F_{kl}^{23} = -f_1^{(a_2)}(q^{m_k^2+\half})  \cdot q^{-\half (a_3+1)(n_l^3 +\half)^2} \cdot G_{lk}^{32},\\
& F_{kl}^{32} = (-1)^{m_k^3} \cdot q^{\half (a_3+1)(m_k^3 +\half)^2}\cdot q^{n_l^2 +\half}
 f_2^{(a_2)}(q^{n_l^2 +\half}) \cdot G_{lk}^{23}.\\
\end{split}
\end{equation*}
From the above relations between entries of $(\tF_{kl})$ and $(G_{kl})$,
we easily see that for an arbitrary permutation $\sigma \in S_r$ we always have:
\begin{equation*}
\begin{split}
\prod_{k=1}^r  \tF_{k\sigma(k)}
=& (-1)^{\sum_{k=1}^{r_1} n_k^1 + \sum_{k=1}^{r_3} m_k^3 +  r_2 }
 q^{a_1\kappa_{\mu^1} /2 + (a_3 +1)\kappa_{\mu^3} /2  } \\
 & \cdot \prod_{k=1}^{r_2} f_1^{(a)} (q^{m_k^2+\half}) \cdot
\prod_{k=1}^{r_2} \big( q^{n_k^2+\half} \cdot f_2^{(a_2)} (q^{n_k^2+\half}) \big)
\cdot \prod_{k=1}^r  G_{k\sigma(k)} \\
=& (-1)^{\sum_{k=1}^{r_1} n_k^1 + \sum_{k=1}^{r_3} m_k^3 +  r_2 }
 q^{\frac{a_1}{2}\kappa_{\mu^1}  + \frac{a_3 +1}{2}\kappa_{\mu^3}  }
 \cdot
 c_{\mu^2}^{(a_2)}(q) \cdot \prod_{k=1}^r  G_{k\sigma(k)},
\end{split}
\end{equation*}
and then
\begin{equation*}
\begin{split}
\det(\tF_{kl}) =& \sum_{\sigma \in S_r}
 (-1)^\sigma \prod_{k=1}^r  \tF_{k\sigma(k)} \\
 =&
 (-1)^{\sum_{k=1}^{r_1} n_k^1 + \sum_{k=1}^{r_3} m_k^3 +  r_2 }
 q^{\frac{a_1}{2}\kappa_{\mu^1}  + \frac{a_3 +1}{2}\kappa_{\mu^3}   }
 c_{\mu^2}^{(a_2)}(q)  \sum_{\sigma \in S_r}
 (-1)^\sigma
 \prod_{k=1}^r  G_{k\sigma(k)} \\
 =& (-1)^{\sum_{k=1}^{r_1} n_k^1 + \sum_{k=1}^{r_3} m_k^3 +  r_2 }
 q^{\frac{a_1}{2}\kappa_{\mu^1}  + \frac{a_3 +1}{2}\kappa_{\mu^3}   }
 c_{\mu^2}^{(a_2)}(q)  \cdot \det(G_{kl}).
\end{split}
\end{equation*}
Therefore \eqref{eq-det-toprove-1} is proved.
\end{proof}

\subsection{Equivalence of two combinatorial expressions}
\label{sec-equiv-vertex}

In this subsection,
we show that $C_{\mu^1,\mu^2,\mu^3}^{(a_1,a_2,a_3)} (q)$ matches with
the framed topological vertex $W_{\mu^1,\mu^2,\mu^3}^{(a_1,a_2,a_3)} (q)$
defined by \eqref{eq-framedTV-def}:
\begin{Theorem}
[=Theorem \ref{thm:Main1}]
\label{thm-equiv-twovert}
For every $(a_1,a_2,a_3)\in \bZ^3$ and partitions $\mu^1,\mu^2,\mu^3$,
we have:
\be
C_{\mu^1,\mu^2,\mu^3}^{(a_1,a_2,a_3)} (q) = W_{\mu^1,\mu^2,\mu^3}^{(a_1,a_2,a_3)} (q).
\ee
\end{Theorem}

To prove this theorem,
we need the following:
\begin{Lemma}
\label{lem-equiv-tGa}
We have:
\be
\label{eq-lem-equiv1}
\begin{split}
&\tGa_+^{(a)}(z) \Gamma_-(w) = \frac{1}{1-zw} \Gamma_-(w) \tGa_+^{(a)}(z),\\
&\Gamma_+(z) \tGa_-^{(a)}(w) = \frac{1}{1-zw} \tGa_-^{(a)}(w) \Gamma_+(z),
\end{split}
\ee
and
\be
\label{eq-lem-equiv2}
\tGa_+^{(a)}(z) \tGa_-^{(a)}(w) = \frac{zw}{1-zw} \tGa_-^{(a)}(w) \tGa_+^{(a)}(z).
\ee
\end{Lemma}
\begin{proof}
First recall that by \eqref{eq-mod-vert} we have
\begin{equation*}
\tGa_+^{(a)} (z) = f_1^{(a)}(z) \Gamma_-(z^{-1}) \psi^*(z^{-1}),
\end{equation*}
and thus
\be
\label{eq-eqvui-pf1}
\begin{split}
\tGa_+^{(a)} (z) \Gamma_-(w)
=& f_1^{(a)}(z) \Gamma_-(z^{-1}) \psi^*(z^{-1}) \Gamma_-(w) \\
=& f_1^{(a)}(z) \Gamma_-(z^{-1}) \Gamma_-(w)\Gamma_-(w)^{-1}\psi^*(z^{-1}) \Gamma_-(w)\\
=& \Gamma_-(w) \cdot f_1^{(a)}(z) \Gamma_-(z^{-1}) \cdot
\big(\Gamma_-(w)^{-1}\psi^*(z^{-1}) \Gamma_-(w)\big).
\end{split}
\ee
Now we compute the conjugate $\Gamma_-(w)^{-1}\psi^*(z^{-1}) \Gamma_-(w)$.
First notice that by the commutation relation \eqref{eq-comm-alphapsi} we have:
\begin{equation*}
\begin{split}
\big[\sum_{n\geq 1} \frac{w^n}{n}\alpha_{-n}, \psi^*(z^{-1}) \big]
=& -\sum_{n\geq 1} \sum_{r \in \bZ+\half} \frac{w^n}{n} \psi_{-n+r}^* z^{r+\half}\\
=& -\sum_{n\geq 1} \frac{z^n w^n}{n} \cdot \psi^* (z^{-1})\\
=& \log(1-zw) \cdot \psi^*(z^{-1}),
\end{split}
\end{equation*}
and then by the Baker-Campbell-Hausdorff formula we have:
\begin{equation*}
\begin{split}
\Gamma_-(w)^{-1}\psi^*(z^{-1}) \Gamma_-(w)
=& \psi^*(z^{-1}) - [\sum_{n\geq 1} \frac{w^n}{n}\alpha_{-n}, \psi^*(z^{-1})]
+ \cdots\\
=& \exp\big( -\log(1-zw) \big) \cdot \psi^*(z^{-1})\\
=& \frac{1}{1-zw}\psi^*(z^{-1}).
\end{split}
\end{equation*}
Now plug this relation into \eqref{eq-eqvui-pf1},
and in this way we obtain the first identity in \eqref{eq-lem-equiv1}.
The second identity in \eqref{eq-lem-equiv1} can be proved using the same method
and thus we omit the details.

Now we prove \eqref{eq-lem-equiv2}.
By \eqref{eq-def-tGamma} and \eqref{eq-comm-gammapm} we have:
\begin{equation*}
\begin{split}
\tGa_+^{(a)}(z) \tGa_-^{(a)}(w) =& f_1^{(a)}(z) f_2^{(a)}(w) \cdot
\Gamma_+(z) R^{-1} z^C w^C R \Gamma_-(w) \\
=& f_1^{(a)}(z) f_2^{(a)}(w)\cdot (zw)^{C+1} \Gamma_+(z)\Gamma_-(w) \\
=& \frac{zw}{1-zw} \cdot f_1^{(a)}(z) f_2^{(a)}(w) \cdot (zw)^C\Gamma_-(w)\Gamma_+(z).
\end{split}
\end{equation*}
Moreover,
it is easy to check that $R \alpha_m R^{-1} = \alpha_m$ for every $m\not=0$,
and thus we have $R\Gamma_\pm (z) R^{-1}= \Gamma_\pm (z)$.
Then:
\begin{equation*}
\begin{split}
\tGa_-^{(a)}(w) \tGa_+^{(a)}(z) =& f_1^{(a)}(z) f_2^{(a)}(w) \cdot
w^C R \Gamma_-(w) \Gamma_+(z) R^{-1}z^C \\
=& f_1^{(a)}(z) f_2^{(a)}(w) \cdot (zw)^C \Gamma_-(w) \Gamma_+(z),
\end{split}
\end{equation*}
and then \eqref{eq-lem-equiv2} follows.
\end{proof}

Now using this lemma,
we rewrite the operator $\Psi_\mu^{(a)}$ in Definition \ref{def-fermvert} as follows:
\begin{Proposition}
\label{prop-rewrite-gammamu}
Let $a\in \bZ$ and let $\mu = (m_1,\cdots,m_r| n_1,\cdots,n_r)$ be a partition.
Then:
\be
\begin{split}
\Psi_\mu^{(a)} (q) =&
z^{2rC} \cdot
q^{a\kappa_\mu /2}\cdot s_\mu (q^\rho) \cdot
\prod_{\substack{j\geq 0 \\ j\not= m_k,\forall k}} \Gamma_- (q^{-j-\half}) \cdot
\prod_{k=1}^r \Gamma_- (q^{n_k+\half})\\
& \cdot  \prod_{k=1}^r \Gamma_+ (q^{m_k+\half}) \cdot
\prod_{\substack{j\geq 0 \\ j\not= n_k,\forall k}} \Gamma_+ (q^{-j-\half}).
\end{split}
\ee
\end{Proposition}
\begin{proof}
First consider the case of a hook partition $\mu = (m|n)$.
By the definition of $\Psi_{\mu}^{(a)} (q)$ we know that:
\begin{equation*}
\begin{split}
\Psi_{(m|n)}^{(a)} (q) =& \prod_{j= m+1}^\infty \Gamma_{-}(q^{-j-\half})
\cdot \tGa_+^{(a)} (q^{m+\half}) \cdot \prod_{j= 0}^{m-1} \Gamma_{-}(q^{-j-\half})\\
& \cdot \prod_{j= 0}^{n-1} \Gamma_{+}(q^{-j-\half}) \cdot \tGa_-^{(a)} (q^{n+\half})
\cdot \prod_{j= n+1}^\infty \Gamma_{+}(q^{-j-\half}).
\end{split}
\end{equation*}
Now using Lemma \ref{lem-equiv-tGa}, we can rewrite it as:
\begin{equation*}
\begin{split}
\Psi_{(m|n)}^{(a)} (q) =&
\frac{1}{ \prod_{i=1}^m (1-q^i)} \cdot
\prod_{j\geq 0 , j\not= m} \Gamma_{-}(q^{-j-\half})
\cdot \tGa_+^{(a)} (q^{m+\half}) \\
& \cdot \frac{1}{ \prod_{i=1}^n (1-q^i)} \cdot
 \tGa_-^{(a)} (q^{n+\half})
\cdot \prod_{j\geq 0,j\not= n} \Gamma_{+}(q^{-j-\half})\\
=& \frac{q^{m+n+1}}{1-q^{m+n+1}} \cdot \frac{1}{ \prod_{i=1}^m (1-q^i)}
\cdot \frac{1}{ \prod_{i=1}^n (1-q^i)} \cdot \\
& \prod_{j\geq 0 , j\not= m} \Gamma_{-}(q^{-j-\half})
\cdot \tGa_-^{(a)} (q^{n+\half}) \tGa_+^{(a)} (q^{m+\half}) \cdot
\prod_{j\geq 0,j\not= n} \Gamma_{+}(q^{-j-\half}).
\end{split}
\end{equation*}
And by plugging the definition \eqref{eq-def-tGamma} of $\tGa_\pm^{(a)}$
into this formula,
we obtain:
\begin{equation*}
\begin{split}
\Psi_{(m|n)}^{(a)} (q)
=& z^{2C} f_1^{(a)}(q^{m+\half}) f_2^{(a)}(q^{n+\half}) \cdot
\frac{q^{m+n+1}}{1-q^{m+n+1}}  \frac{1}{ \prod_{i=1}^m (1-q^i)}
 \frac{1}{ \prod_{i=1}^n (1-q^i)} \cdot \\
& \prod_{j\geq 0 , j\not= m} \Gamma_{-}(q^{-j-\half})
\cdot \Gamma_- (q^{n+\half}) \Gamma_+ (q^{m+\half}) \cdot
\prod_{j\geq 0,j\not= n} \Gamma_{+}(q^{-j-\half}),
\end{split}
\end{equation*}
where we have also used the facts that the operators $\Gamma_\pm (z)$ are of charge $0$,
and $R\Gamma_\pm(z) R^{-1} = \Gamma_\pm(z)$.
Recall that $f_1^{(a)}$ and $f_2^{(a)}$ are defined by \eqref{eq-def-f1f2},
and thus:
\begin{equation*}
\begin{split}
\Psi_{(m|n)}^{(a)} (q)
=& z^{2C} \cdot
\frac{(-1)^{m+n+1} \cdot q^{\half a (m^2+m-n^2-n) +\half m^2 +m +\half n +\half} }
{(1-q^{m+n+1}) \cdot \prod_{i=1}^m (1-q^i) \cdot \prod_{i=1}^n (1-q^i)} \cdot \\
& \prod_{j\geq 0 , j\not= m} \Gamma_{-}(q^{-j-\half})
\cdot \Gamma_- (q^{n+\half}) \Gamma_+ (q^{m+\half}) \cdot
\prod_{j\geq 0,j\not= n} \Gamma_{+}(q^{-j-\half}).
\end{split}
\end{equation*}
On the other hand,
using \eqref{eq-eval-schur-rho} we know that:
\begin{equation*}
\begin{split}
&q^{a\kappa_{(m|n)}/2}s_{(m|n)}(q^\rho)\\
=& q^{(\frac{a}{2}+\frac{1}{4})(m^2+m-n^2-n)}
\frac{1}{q^{\frac{m+n+1}{2}} - q^{-\frac{m+n+1}{2}}}
\cdot \frac{1}{\prod_{i=1}^m (q^{\frac{i}{2}} - q^{-\frac{i}{2}})} \cdot
 \frac{1}{\prod_{i=1}^n (q^{\frac{i}{2}} - q^{-\frac{i}{2}})}\\
 =& (-1)^{m+n+1} \cdot
 \frac{ q^{\frac{a}{2} (m^2+m-n^2-n)} \cdot q^{ \half m^2 +m +\half n +\half} }
{(1-q^{m+n+1}) \cdot \prod_{i=1}^m (1-q^i) \cdot \prod_{i=1}^n (1-q^i)},
\end{split}
\end{equation*}
and thus the proposition holds for the hook partition $\mu = (m|n)$.

Now we consider the general case $\mu = (m_1,\cdots,m_r| n_1,\cdots,n_r)$.
Recall that:
\begin{equation*}
\Psi_{\mu}^{(a)} (q) = \prod_{j\geq 0}^{\longleftarrow} \Gamma_{-,\{j,\mu\}}^{(a)}(q^{-j-\half})
\cdot \prod_{i\geq 0}^{\longrightarrow} \Gamma_{+,\{i,\mu\}}^{(a)}(q^{-i-\half}),
\end{equation*}
where each factor $\Gamma_{\pm,\{j,\mu\}}^{(a)}(q^{-j-\half})$ is
either $\Gamma_\pm (q^{-j-\half})$ or $\tGa_\mp^{(a)}(q^{j+\half})$.
By applying Lemma \ref{lem-equiv-tGa},
we see that the operator $\Gamma_{\mu}^{(a)} (q)$ is of the form:
\begin{equation*}
\begin{split}
\Psi_\mu^{(a)} (q) =& \tc(a,\mu;q) \cdot
\prod_{\substack{j\geq 0 \\ j\not= m_k,\forall k}} \Gamma_- (q^{-j-\half}) \cdot
\prod_{k=1}^r \tGa_-^{(a)} (q^{n_k+\half})\\
& \cdot  \prod_{k=1}^r \tGa_+^{(a)} (q^{m_k+\half}) \cdot
\prod_{\substack{j\geq 0 \\ j\not= n_k,\forall k}} \Gamma_+ (q^{-j-\half}),
\end{split}
\end{equation*}
and then:
\begin{equation*}
\begin{split}
\Psi_\mu^{(a)} (q) =& z^{2rC} \cdot c(a,\mu;q) \cdot
\prod_{\substack{j\geq 0 \\ j\not= m_k,\forall k}} \Gamma_- (q^{-j-\half}) \cdot
\prod_{k=1}^r \Gamma_- (q^{n_k+\half})\\
& \cdot  \prod_{k=1}^r \Gamma_+ (q^{m_k+\half}) \cdot
\prod_{\substack{j\geq 0 \\ j\not= n_k,\forall k}} \Gamma_+ (q^{-j-\half}),
\end{split}
\end{equation*}
again by the facts that $\Gamma_\pm (z)$ are of charge $0$
and $R\Gamma_\pm(z) R^{-1} = \Gamma_\pm(z)$.
Recall that we have
$\langle 0| \Gamma_-(z) = \langle 0| $ and $\Gamma_+ (z)|0\rangle = |0\rangle$,
and thus the coefficient $c(a,\mu;q)$ is
the vacuum expectation value of the operator $\Psi_\mu^{(a)} (q)$:
\begin{equation*}
c(a,\mu;q) = \langle \Psi_\mu^{(a)} (q) \rangle.
\end{equation*}
Therefore it is sufficient to show that:
\be
\langle \Psi_\mu^{(a)} (q) \rangle = q^{a\kappa_\mu /2}\cdot s_\mu (q^\rho).
\ee

Now we need to compute the vacuum expectation value $\langle \Psi_\mu^{(a)} (q) \rangle$.
This has been done in the proof of Proposition \ref{sec-det-C} using Wick's Theorem.
In fact,
it is clear that:
\begin{equation*}
\langle \Psi_\mu^{(a)} (q) \rangle = C_{(\emptyset),\mu,(\emptyset)}^{(0,a,0)}(q).
\end{equation*}
Thus we may apply the determinantal formula given in Proposition \ref{sec-det-C} and obtain:
\be
\langle \Psi_\mu^{(a)} (q) \rangle =
\det \big( \langle \Psi_{(m_i | n_j)}^{(a)} (q) \rangle \big)_{1\leq i,j \leq r}.
\ee
Notice that we have already proved the result for the hook cases $(m_i|n_j)$,
i.e.,
the entries of the above matrix are given by:
\begin{equation*}
\langle \Psi_{(m_i | n_j)}^{(a)} (q) \rangle
= q^{a\kappa_{(m_i|n_j)} /2}\cdot s_{(m_i|n_j)} (q^\rho).
\end{equation*}
Therefore,
now it is sufficient to show that:
\begin{equation*}
q^{a\kappa_\mu /2}\cdot s_\mu (q^\rho) \cdot =
\det\big( q^{a\kappa_{(m_i|n_j)} /2}\cdot s_{(m_i|n_j)} (q^\rho) \big)_{1\leq i,j\leq r}
\end{equation*}
for a partition $\mu = (m_1,\cdots,m_r| n_1,\cdots,n_r)$.
This is actually a straightforward consequence of the Giambelli formula
(see e.g. \cite[\S I.3, Example 9]{mac})
\begin{equation*}
s_\mu = \det (s_{(m_i|n_j)})_{1\leq i,j\leq r}
\end{equation*}
together with the definition $\kappa_\mu = \sum_{i=1}^r m_i(m_i+1) - \sum_{j=1}^r n_j(n_j+1)$.
This completes the proof.
\end{proof}

Now we are able to prove Theorem \ref{thm-equiv-twovert}.
\begin{proof}
[Proof of Theorem \ref{thm-equiv-twovert}]

Let $\mu = (\mu_1,\cdots,\mu_l) = (m_1,\cdots,m_k| n_1,\cdots,n_k)$ be a partition,
and denote $\rho = (-\half, -\frac{3}{2} , -\frac{5}{2},\cdots)$.
Fix such a partition $\mu$,
and denote by $\bm x =  (x_1,x_2, x_3,\cdots)$ the sequence:
\begin{equation*}
x_i = q^{\mu_i -i +\half},
\end{equation*}
i.e., $\bm x$ is the sequence
\begin{equation*}
\bm x = q^{\mu+\rho} =
( q^{\mu_1 - \half}, q^{\mu_2-\frac{3}{2}}, \cdots, q^{\mu_l - l+\half},
q^{ -l-\half}, q^{-l-\frac{3}{2}},\cdots ).
\end{equation*}
Denote by $\bm t(\mu)  = (t_1,t_2,t_3,\cdots)$ the sequence
\begin{equation*}
t_n = \frac{1}{n}  p_n(\bm x),
\end{equation*}
where $p_n$ is the Newton symmetric function of degree $n$.
Then:
\be
\label{eq-pfequiv-infprod}
\begin{split}
\Gamma_\pm( \bm t (\mu))
= \exp\Big( \sum_{n\geq 0} \frac{1}{n}   p_n (\bm x) \alpha_{\pm n} \Big)
= \exp\Big( \sum_{n\geq 0}  \frac{1}{n}  \sum_{j\geq 0}  (x_j)^n \alpha_{\pm n} \Big)
= \prod_{j\geq 0} \Gamma_\pm (x_j),
\end{split}
\ee
where $\Gamma_\pm( \bm t (\mu)) $ means plugging the sequence $\bm t(\mu)$
into the definition \eqref{e-def-gammat}.
Moreover,
notice that
\begin{equation*}
\begin{split}
&\{\mu+\rho\}
= \big\{m_j+\half \big\}_{j=1}^k \bigsqcup
\Big( \big(\bZ_{<0}+\half\big) \big\backslash \{-n_j-\half\}_{j=1}^n\Big),
\\
&\{\mu^t +\rho\} =
 \big\{n_j+\half \big\}_{j=1}^k \bigsqcup
\Big( \big(\bZ_{<0}+\half\big) \big\backslash \{-m_j-\half\}_{j=1}^n\Big),
\end{split}
\end{equation*}
as sets,
and thus the conclusion of Proposition \ref{prop-rewrite-gammamu} can be rewritten as:
\be
\Psi_\mu^{(a)} = z^{2rC} \cdot q^{a\kappa_\mu /2} s_\mu(q^\rho)
\cdot \Gamma_- (\bm t(\mu^t)) \Gamma_+(\bm t(\mu)).
\ee
Then the combinatorial expression $C_{\mu^1,\mu^2,\mu^3}^{(a_1,a_2,a_3)} (q)$ can be rewritten as:
\begin{equation*}
\begin{split}
C_{\mu^1,\mu^2,\mu^3}^{(a_1,a_2,a_3)} (q) =&
\langle \mu^1 | q^{a_1 K/2} \Psi_{\mu^2}^{(a_2)}(q) q^{-(a_3 +1)K /2} |(\mu^3)^t \rangle \\
=& q^{a_1 \kappa_{\mu^1}/2 + (a_3+1)\kappa_{\mu^3}/2}  \cdot
\langle \mu^1 | \Psi_{\mu^2}^{(a_2)}(q) | (\mu^3)^t \rangle \\
=& q^{a_1 \kappa_{\mu^1}/2 + a_2 \kappa_{\mu^2}/2+ (a_3+1)\kappa_{\mu^3}/2}
s_{\mu^2}(q^\rho)  \cdot
\langle \mu^1 |  \Gamma_- \big(\bm t((\mu^2)^t) \big)
\Gamma_+ \big(\bm t(\mu^2) \big) | (\mu^3 )^t\rangle.
\end{split}
\end{equation*}
Using \eqref{eq-eval-schur-rho} one may easily find that:
\begin{equation*}
s_{\mu^2} (q^\rho)
= (-1)^{|\mu^2|} s_{(\mu^2)^t} (q^{-\rho}),
\end{equation*}
and thus:
\begin{equation*}
\begin{split}
 C_{\mu^1,\mu^2,\mu^3}^{(a_1,a_2,a_3)} (q)
 =&
 (-1)^{|\mu^2|}  q^{a_1 \kappa_{\mu^1}/2 + a_2 \kappa_{\mu^2}/2+ (a_3+1)\kappa_{\mu^3}/2}
s_{(\mu^2)^t}(q^{-\rho}) \\
&\cdot \langle \mu^1 |  \Gamma_- \big(\bm t((\mu^2)^t) \big)
\Gamma_+ \big(\bm t(\mu^2) \big) | (\mu^3 )^t\rangle.
\end{split}
\end{equation*}
Furthermore,
by \eqref{eq-special-schur} and \eqref{eq-special-schur-2} we know that
\begin{equation*}
\langle \mu^1 |  \Gamma_- \big(\bm t((\mu^2)^t) \big)
\Gamma_+ \big(\bm t(\mu^2) \big) | (\mu^3 )^t\rangle
= \sum_\eta s_{\mu^1 / \eta} (q^{(\mu^2)^t+\rho}) \cdot  s_{(\mu^3)^t / \eta} (q^{\mu^2+\rho}),
\end{equation*}
and now we can conclude that:
\begin{equation*}
\begin{split}
 C_{\mu^1,\mu^2,\mu^3}^{(a_1,a_2,a_3)} (q)
 =&
 (-1)^{|\mu^2|} \cdot q^{a_1 \kappa_{\mu^1}/2 + a_2 \kappa_{\mu^2}/2+ (a_3+1)\kappa_{\mu^3}/2} \cdot
s_{(\mu^2)^t}(q^{-\rho}) \\
& \cdot
\sum_\eta s_{\mu^1 / \eta} (q^{(\mu^2)^t+\rho}) \cdot  s_{(\mu^3)^t / \eta} (q^{\mu^2+\rho}).
\end{split}
\end{equation*}
Now compare this with \eqref{eq-framedTV-def},
then we get the conclusion of Theorem \ref{thm-equiv-twovert}.
\end{proof}

\subsection{Computing the entries $F_{mn}^{ij}$}
\label{eq-computeF}

In this subsection,
we give the explicit formulas of the entries $\tF_{kl}(q;\bm a)$
(or equivalently, of $F_{mn}^{ij}$) in the above determinantal formula
for $C_{\mu^1,\mu^2,\mu^3}^{(a_1,a_2,a_3)} (q)$.

We will need the following:
\begin{Lemma}
\label{lem-skewS-eval}
We have:
\be
\begin{split}
& s_{(m)}(q^{(n)+\rho})
= \sum_{k=0}^{\min (m,n)}
q^{-\half k^2 -\half k +\half km +\half kn +\frac{1}{4}m^2 -\frac{1}{4} m}
\cdot  \frac{[n]!}{[n-k]!\cdot [m-k]!},\\
& s_{(1^n)} (q^{(1^m)+\rho})
= \sum_{k=0}^{\min (m,n)}
q^{\half k^2 +\half k -\half km -\half kn -\frac{1}{4}n^2 +\frac{1}{4} n}
\cdot  \frac{[m]!}{[n-k]!\cdot [m-k]!}.
\end{split}
\ee
\end{Lemma}
\begin{proof}
We will use the following identity
(cf. \cite[(29)]{zhou5} and \cite[(20)]{zhou4}):
\begin{equation*}
s_\nu (q^{\mu +\rho}) = (-1)^{|\nu|} q^{\kappa_\nu /2}
\sum_\eta \frac{s_{\mu/\eta} (q^{-\rho})}{s_\mu(q^{-\rho})} s_{\nu/\eta} (q^{-\rho})
\end{equation*}
Now we take $\mu =(n)$ and $\nu = (m-1|0) = (m)$ in this identity
and then apply \eqref{eq-skewS-length1},
and in this way we obtain:
\begin{equation*}
\begin{split}
s_{(m)}(q^{(n)+\rho}) =& (-1)^m q^{\kappa_{(m)}/2} \cdot
\sum_{\eta} \frac{s_{(n)/\eta}(q^{-\rho})}{s_{(n)}(q^{-\rho})}
s_{(m)/\eta}(q^{-\rho}) \\
=& (-1)^m q^{\kappa_{(m)}/2} \cdot
\sum_{k=0}^{\min (m,n)} \frac{s_{(n-k)}(q^{-\rho})}{s_{(n)}(q^{-\rho})}
s_{(m-k)}(q^{-\rho})
\end{split}
\end{equation*}
From \eqref{eq-eval-schur-rho} we know that:
\begin{equation*}
s_{(i)} (q^{-\rho}) = q^{-\kappa_{(i)}/4}\cdot
\frac{1}{ \prod_{j=1}^i (-[j]) } =
(-1)^i q^{-i(i-1)/4}\cdot
\frac{1}{ [i]! },
\end{equation*}
and thus:
\begin{equation*}
\begin{split}
s_{(m)}(q^{(n)+\rho})
=&  q^{m(m-1)/2}
\sum_{k=0}^{\min (m,n)}
\frac{q^{-(n-k)(n-k-1)/4 - (m-k)(m-k-1)/4 } \cdot [n]! }{q^{n(n-1)/4} \cdot [n-k]!\cdot[m-k]!} \\
=& \sum_{k=0}^{\min (m,n)}
q^{-\half k^2 -\half k +\half km +\half kn +\frac{1}{4}m^2 -\frac{1}{4} m}
\cdot  \frac{[n]!}{[n-k]!\cdot [m-k]!}.
\end{split}
\end{equation*}
This proves the first equality in this lemma.
The second equality can be proved in the same way and thus we omit the details.
\end{proof}

\begin{Proposition}
\label{prop-coeff-F}
We have:
\be
\begin{split}
& F_{mn}^{ii}(q;\bm a)
= q^{(2a_i +1)(m^2+m-n^2-n)/4} \cdot \frac{1}{[m+n+1] \cdot [m]! \cdot [n]!},
\qquad i=1,2,3;
\end{split}
\ee
and:
\be
\begin{split}
&F^{12}_{mn}(q;\bm a)
= q^{\frac{1}{4}(2a_1+1)(m+\half)^2  - \frac{1}{4}(2a_2+1)(n+\half)^2 }
\sum_{k=0}^{\min (m,n)}
\frac{ q^{\frac{1}{2}(k+1)(m+n-k) + \frac{3}{16}} }{[n-k]!\cdot [m-k]!},\\
&F^{23}_{mn}(q;\bm a) =
q^{\frac{1}{4}(2a_2+1)(m+\half)^2  - \frac{1}{4}(2a_3+1)(n+\half)^2 }
\sum_{k=0}^{\min (m,n)}
\frac{ q^{\frac{1}{2}(k+1)(m+n-k) + \frac{3}{16}} }{[n-k]!\cdot [m-k]!}, \\
&F^{31}_{mn}(q;\bm a) =
q^{\frac{1}{4}(2a_3+1)(m+\half)^2  - \frac{1}{4}(2a_1+1)(n+\half)^2 }
\sum_{k=0}^{\min (m,n)}
\frac{ q^{\frac{1}{2}(k+1)(m+n-k) + \frac{1}{8}} }{[n-k]!\cdot [m-k]!},
\end{split}
\ee
and
\be
\label{eq-formula-f13}
\begin{split}
&F^{21}_{mn}(q;\bm a) = -
q^{\frac{1}{4}(2a_2+1)(m+\half)^2  - \frac{1}{4}(2a_1+1)(n+\half)^2 }
\sum_{k=0}^{\min (m,n)}
\frac{ q^{-\frac{1}{2}(k+1)(m+n-k) - \frac{3}{16}} }{[n-k]!\cdot [m-k]!}, \\
&F^{32}_{mn}(q;\bm a) = -
q^{\frac{1}{4}(2a_3+1)(m+\half)^2  - \frac{1}{4}(2a_2+1)(n+\half)^2 }
\sum_{k=0}^{\min (m,n)}
\frac{ q^{-\frac{1}{2}(k+1)(m+n-k) - \frac{3}{16}} }{[n-k]!\cdot [m-k]!}, \\
&F^{13}_{mn}(q;\bm a) = -
q^{\frac{1}{4}(2a_1+1)(m+\half)^2  - \frac{1}{4}(2a_3+1)(n+\half)^2 }
\sum_{k=0}^{\min (m,n)}
\frac{ q^{-\frac{1}{2}(k+1)(m+n-k) - \frac{1}{8}} }{[n-k]!\cdot [m-k]!}.
\end{split}
\ee
\end{Proposition}

\begin{proof}
First consider the case $(i,j) = (1,1)$.
We have:
\begin{equation*}
\begin{split}
F_{mn}^{11}(q;\bm a) = &
\langle (m|n) | q^{a_1 K/2} \Psi_\emptyset^{(a_2)} (q) |0\rangle \\
= &  q^{ a_1 \kappa_{(m|n)}/2}\cdot
\langle (m|n) | \Psi_\emptyset^{(a_2)} (q) |0\rangle \\
=&  q^{ a_1 \kappa_{(m|n)}/2}\cdot \
\langle (m|n) |
\prod_{j\geq 0} \Gamma_-(q^{-j-\half}) |0\rangle.
\end{split}
\end{equation*}
By taking $\mu = (\emptyset)$ in \eqref{eq-pfequiv-infprod}
and then applying \eqref{eq-special-schur-2} and \eqref{eq-eval-schur-rho},
we know that:
\begin{equation*}
\langle (m|n) |
\prod_{j\geq 0} \Gamma_-(q^{-j-\half}) |0\rangle
=  s_{(m|n)} (q^\rho)
= q^{\kappa_{(m|n)}/4} \cdot \frac{1}{\prod_{e\in (m|n)} [h(e)]},
\end{equation*}
and thus:
\begin{equation*}
\begin{split}
F_{mn}^{11}(q;\bm a) = &
q^{(2a_1 +1)\kappa_{(m|n)}/4} \cdot  \frac{1}{\prod_{e\in (m|n)} [h(e)]}\\
=& q^{(2a_1 +1)(m^2+m-n^2-n)/4} \cdot \frac{1}{[m+n+1] \cdot [m]! \cdot [n]!}.
\end{split}
\end{equation*}
Now consider the case $(i,j) = (2,2)$.
By Proposition \ref{prop-rewrite-gammamu} and \eqref{eq-eval-schur-rho} we have:
\begin{equation*}
\begin{split}
F^{22}_{mn} (q;\bm a)=&
\langle  \Psi_{(m|n)}^{(a_2)} (q) \rangle \\
=& q^{a_2 \kappa_{(m|n)}/2} s_{(m|n)} (q^\rho) \\
=& q^{(2a_2 +1)(m^2+m-n^2-n)/4} \cdot \frac{1}{[m+n+1] \cdot [m]! \cdot [n]!}.
\end{split}
\end{equation*}
And for $(i,j)=(3,3)$, we have:
\begin{equation*}
\begin{split}
F^{33}_{mn} (q;\bm a) =&
\langle 0 | \Psi_{\emptyset}^{(a_2)} (q) q^{-(a_3+1)K/2} |(n|m)\rangle\\
=& q^{-(a_3+1) \kappa_{(n|m)} /2} \cdot \langle 0| \Psi_{\emptyset}^{(a_2)} (q)|(n|m)\rangle \\
=&  q^{-(a_3+1) \kappa_{(n|m)} /2} \cdot s_{(n|m)} (q^\rho) \\
=&  q^{(2a_3+1) (m^2+m-n^2-n) /4} \cdot \frac{1}{[m+n+1] \cdot [m]! \cdot [n]!}.
\end{split}
\end{equation*}

Now consider the case $(i,j)=(1,2)$.
This case is a little more complicated than the above three cases.
By \eqref{eq-comm-expKpsi} and $\Gamma_+(z)|0\rangle = |0\rangle$ we have:
\begin{equation*}
\begin{split}
F^{12}_{mn}(q;\bm a) =&
-\langle  \psi^*_{m+\half} q^{a_1 K /2}\Psi_{\emptyset,n}^{(a_2)} (q) \rangle \\
=& -q^{a_1 (m+\half)^2 /2}
\langle  \psi^*_{m+\half} \Psi_{\emptyset,n}^{(a_2)} (q) \rangle \\
=& -q^{a_1 (m+\half)^2 /2}
\langle  \psi^*_{m+\half}
\prod_{j=0}^\infty \Gamma_- (q^{-j-\half}) \cdot
\prod_{j=0}^{n-1}\Gamma_+ (q^{-j-\half}) \cdot \tGa_-^{(a_2)} (q^{n+\half})
 \rangle,
\end{split}
\end{equation*}
and then by Lemma \ref{lem-equiv-tGa} and \eqref{eq-def-tGamma}:
\begin{equation*}
\begin{split}
F^{12}_{mn}(q;\bm a) =&
 -\frac{q^{a_1 (m+\half)^2 /2}} {\prod_{j=0}^{n-1} (1-q^{n-j})}
\langle  \psi^*_{m+\half}
\prod_{j=0}^\infty \Gamma_- (q^{-j-\half}) \cdot
 \tGa_-^{(a_2)} (q^{n+\half})
 \rangle \\
=& -\frac{q^{a_1 (m+\half)^2 /2} f_2^{(a_2)}(q^{n+\half})  q^{n+\half}}
{\prod_{j=0}^{n-1} (1-q^{n-j})}
 \langle \psi^*_{m+\half}
\prod_{j=0}^\infty \Gamma_- (q^{-j-\half}) \cdot
 R\Gamma_- (q^{n+\half})
 \rangle\\
=& \frac{q^{\half a_1 (m+\half)^2 -\half a_2 (n+\half)^2 +\frac{n}{4}-\frac{n^2}{4}+\frac{3}{16}}}{[n]!}
\langle \psi^*_{m+\half}
\prod_{j=0}^\infty \Gamma_- (q^{-j-\half}) \cdot
 R\Gamma_- (q^{n+\half})
 \rangle.
\end{split}
\end{equation*}
Recall that $R^{-1}\Gamma_\pm(z) R =\Gamma_\pm(z)$, and thus:
\begin{equation*}
\begin{split}
 \langle \psi^*_{m+\half}
\prod_{j=0}^\infty \Gamma_- (q^{-j-\half}) \cdot
 R\Gamma_- (q^{n+\half})
 \rangle
=& \langle R (R^{-1} \psi^*_{m+\half} R)
\prod_{j=0}^\infty \big(R^{-1}\Gamma_- (q^{-j-\half})R\big) \cdot
\Gamma_- (q^{n+\half})
 \rangle \\
=& \langle \psi_\half \psi_{m-\half}^* \prod_{j=0}^\infty \Gamma_- (q^{-j-\half}) \cdot
 \Gamma_- (q^{n+\half})
 \rangle,
\end{split}
\end{equation*}
and again by \eqref{eq-pfequiv-infprod} and \eqref{eq-special-schur-2}:
\begin{equation*}
\begin{split}
\langle \psi_\half \psi_{m-\half}^* \prod_{j=0}^\infty \Gamma_- (q^{-j-\half}) \cdot
 \Gamma_- (q^{n+\half}) \rangle
=& s_{(m-1 | 0)} (q^{n+\half},q^{ -\half}, q^{-\frac{3}{2}}, q^{-\frac{5}{2}},\cdots)\\
=& q^{m} \cdot s_{(m-1 | 0)} (q^{n-\half}, q^{-\frac{3}{2}}, q^{-\frac{5}{2}},\cdots)\\
=& q^{m} \cdot s_{(m)} (q^{(n) + \rho}).
\end{split}
\end{equation*}
where $q^{(n)+\rho}$ denotes the sequence
$q^{(n)+\rho}= (q^{n-\half}, q^{-\frac{3}{2}}, q^{-\frac{5}{2}},\cdots)$.

Now apply the Lemma \ref{lem-skewS-eval} to above computations of $F^{12}_{mn}(q;\bm a)$,
and then we are able to conclude that:
\begin{equation*}
\begin{split}
F^{12}_{mn}(q;\bm a)
=& \frac{q^{\half a_1 (m+\half)^2 -\half a_2 (n+\half)^2 +\frac{n}{4}-\frac{n^2}{4}+\frac{3}{16}}}{[n]!} \cdot
q^{m} \cdot s_{(m)} (q^{(n) + \rho})\\
=&  \frac{q^{\half a_1 (m+\half)^2 -\half a_2 (n+\half)^2 +\frac{n}{4}-\frac{n^2}{4}+\frac{3}{16}}}{[n]!} \cdot
q^{m} \cdot \\
& \sum_{k=0}^{\min (m,n)}
q^{-\half k^2 -\half k +\half km +\half kn +\frac{1}{4}m^2 -\frac{1}{4} m}
\cdot  \frac{[n]!}{[n-k]!\cdot [m-k]!} \\
=& q^{\frac{1}{4}(2a_1+1)(m+\half)^2  - \frac{1}{4}(2a_2+1)(n+\half)^2 }
\sum_{k=0}^{\min (m,n)}
\frac{ q^{\frac{1}{2}(k+1)(m+n-k) + \frac{3}{16}} }{[n-k]!\cdot [m-k]!}.
\end{split}
\end{equation*}

The cases $(i,j) = (2,1),(2,3),(3,2)$ can all be computed similarly using Lemma \ref{lem-skewS-eval}
and we omit the details.
Now it remains to prove for $(i,j) = (1,3),(3,1)$.
For $(i,j) = (1,3)$,
by \eqref{eq-comm-expKpsi} we have:
\begin{equation*}
\begin{split}
F^{13}_{mn}(q;\bm a) =&
-\langle  \psi^*_{m+\half} q^{a_1 K/2}\Psi_{\emptyset}^{(a_2)} (q)
 q^{-(a_3+1)K/2} \psi_{-n-\half} \rangle \\
 =&
- q^{\half a_1 (m+\half)^2 - \half (a_3 +1)(n+\half)^2} \cdot
\langle  \psi^*_{m+\half} \Psi_{\emptyset}^{(a_2)} (q) \psi_{-n-\half} \rangle.
\end{split}
\end{equation*}
Moreover,
\begin{equation*}
\begin{split}
\langle  \psi^*_{m+\half} \Psi_{\emptyset}^{(a_2)} (q) \psi_{-n-\half} \rangle
=&
\langle  \psi^*_{m+\half} \prod_{j\geq 0} \Gamma_-(q^{-j-\half}) \cdot
\prod_{j\geq 0} \Gamma_+(q^{-j-\half}) \cdot \psi_{-n-\half} \rangle\\
=& \langle  \psi_\half \psi^*_{m+\half} \prod_{j\geq 0} \Gamma_-(q^{-j-\half}) \cdot
\prod_{j\geq 0} \Gamma_+(q^{-j-\half}) \cdot \psi_{-n-\half} \psi_{-\half}^* \rangle\\
=& \sum_\eta s_{(m-1|0) /\eta}(q^\rho) \cdot s_{(n-1|0) /\eta}(q^\rho) \\
=& \sum_{k=0}^{\min(m,n)}  s_{(m-k)}(q^\rho) \cdot s_{(n-k) }(q^\rho),
\end{split}
\end{equation*}
and by \eqref{eq-eval-schur-rho}:
\begin{equation*}
s_{(m-k)}(q^\rho) \cdot s_{(n-k) }(q^\rho)
= \frac{q^{(m-k)(m-k-1)/4+(n-k)(n-k-1)/4}}{[m-k]!\cdot [n-k]!}.
\end{equation*}
Therefore:
\begin{equation*}
\begin{split}
F^{13}_{mn}(q;\bm a) =&
- q^{\half a_1 (m+\half)^2 - \half (a_3 +1)(n+\half)^2} \sum_{k=0}^{\min(m,n)}
\frac{q^{\frac{1}{4}(m-k)(m-k-1)+\frac{1}{4}(n-k)(n-k-1)}}{[m-k]!\cdot [n-k]!}\\
=& -
q^{\frac{1}{4}(2a_1+1)(m+\half)^2  - \frac{1}{4}(2a_3+1)(n+\half)^2 }
\sum_{k=0}^{\min (m,n)}
\frac{ q^{-\frac{1}{2}(k+1)(m+n-k) - \frac{1}{8}} }{[n-k]!\cdot [m-k]!}.
\end{split}
\end{equation*}
The case $(i,j)= (3,1)$ can be proved similarly and we omit the details.
\end{proof}

\section{Proof of the Framed ADKMV Conjecture}
\label{sec:ADKMV-proof}

In \S \ref{sec-ferm-topovert} we have proved Theorem \ref{thm:Main1} and represented
the framed topological vertex
(i.e., the left-hand side of the framed ADKMV Conjecture \eqref{eq-framedADKMV})
as the determinant of a matrix $(\tF_{kl})$.
Now in this section
we deal with the right-hand side of \eqref{eq-framedADKMV}.
We show that the right-hand side can also be
represented as the determinant of a matrix $(B_{kl})$,
and then show that these two determinants are actually equal to each other.
This completes the proof the framed ADKMV Conjecture.

\subsection{Determinantal formula for the Bogoliubov transform}
\label{sec-det-Bog}

Let $\mu^1$, $\mu^2$, $\mu^3$ be three partitions,
and let
\begin{equation*}
\mu^i = ( m_1^i, \cdots, m_{r_i}^i | n_1^i, \cdots, n_{r_i}^i ),
\qquad i=1,2,3,
\end{equation*}
be their Frobenius notations.
Denote by
\begin{equation*}
B_{\mu^1,\mu^2,\mu^3}^{(\bm a)}(q) = \langle \mu^1,\mu^2,\mu^3|
\exp\Big( \sum_{i,j =1}^3 \sum_{m,n\geq 0} A_{mn}^{ij}(q; \bm a)
\psi_{-m-\half}^{i} \psi_{-n-\half}^{j*} \Big)
|0\rangle \otimes |0\rangle \otimes |0\rangle
\end{equation*}
the right-hand side of the framed ADKMV Conjecture \eqref{eq-framedADKMV},
where the coefficients $A_{mn}^{ij}(q;\bm a)$ are given by equation \eqref{eqn:def Aqa}.
Denote $r=r_1+r_2+r_3$,
and we will use the notations $\bar k$, $\bar l$ and $I_1$, $I_2$, $I_3$ defined in subsection \ref{sec-det-C}.
Then:
\begin{Proposition}
\label{prop-det-bog}
$B_{\mu^1,\mu^2,\mu^3}^{(\bm a)}(q)$
satisfies the following determinantal formula:
\be
\label{eq-prop-det-bog}
B_{\mu^1,\mu^2,\mu^3}^{(\bm a)}(q) = \det \big(B_{kl}(q;\bm a) \big)_{1\leq k,l \leq r}.
\ee
where $(B_{kl})$ is a matrix of size $r\times r$, given by:
\be
B_{kl} (q;\bm a) = (-1)^{n_{\bar l}^j} \cdot
A_{m_{\bar k}^i n_{\bar l}^j}^{ij}(q;\bm a),
\qquad \text{ if $k\in I_i$ and $l\in I_j$}
\ee
for $1\leq k,l \leq r$.
\end{Proposition}
\begin{proof}
Notice that the above Bogoliubov transform only involves fermionic creators
which anti-commute with each other.
In particular,
one has
\begin{equation*}
(\psi_{-m-\half}^i)^2 = (\psi_{-m-\half}^{i*})^2 =0,
\qquad \forall m\ge 0 \text{ and } 1\leq i\leq 3.
\end{equation*}
Therefore we can expand the exponential in the following way:
\begin{equation*}
\begin{split}
\exp\Big( \sum_{i,j =1}^3 \sum_{m,n\geq 0} A_{mn}^{ij}(q;\bm a)
\psi_{-m-\half}^i \psi_{-n-\half}^{j*} \Big)
= \prod_{i,j=1}^3 \prod_{m,n\geq 0}
\Big( 1+ A_{mn}^{ij}(q;\bm a) \psi_{-m-\half}^i \psi_{-n-\half}^{j*} \Big),
\end{split}
\end{equation*}
and then by \eqref{eq-ferm-basismu} we have:
\begin{equation*}
\begin{split}
B_{\mu^1,\mu^2,\mu^3}^{(\bm a)}(q)
= & (-1)^{ \sum_{i=1}^3 \sum_{k=1}^{r_i} n_k^i }  \cdot
\langle 0 | \otimes  \langle 0 | \otimes  \langle 0 |
\prod_{j=1}^{r_1} \psi_{n_j^1 +\half}^{1}\psi_{m_j^1 +\half}^{1*} \cdot
 \prod_{j=1}^{r_2} \psi_{n_j^2 +\half}^{2}\psi_{m_j^2 +\half}^{2*} \\
& \cdot \prod_{j=1}^{r_3} \psi_{n_j^3 +\half}^{3}\psi_{m_j^3 +\half}^{3*}
 \cdot  \prod_{i,j=1}^3 \prod_{m,n\geq 0}
\Big( 1+ A_{mn}^{ij}(q;\bm a) \psi_{-m-\half}^i \psi_{-n-\half}^{j*} \Big)
|0\rangle \otimes |0\rangle \otimes |0\rangle.
\end{split}
\end{equation*}
Recall that $\psi_{s}^{i}$ and $\psi_{-s}^{i*}$ are adjoint to each other
with respect to the inner product on the fermionic Fock space
for every half-integer $s$,
and thus the right-hand side of the above formula is simply
$(-1)^{ \sum_{i=1}^3 \sum_{k=1}^{r_i} n_k^i }$ times the coefficient of
\begin{equation*}
\prod_{j=1}^{r_3} \psi_{-m_j^3-\half}^3 \psi_{-n_j^3-\half}^{3*}
\cdot \prod_{j=1}^{r_2} \psi_{-m_j^2-\half}^2 \psi_{-n_j^2-\half}^{2*}
\cdot \prod_{j=1}^{r_1} \psi_{-m_j^1-\half}^1 \psi_{-n_j^1-\half}^{1*}
\end{equation*}
in the expansion of
\begin{equation*}
\prod_{i,j=1}^3 \prod_{m,n\geq 0}
\Big( 1+ A_{mn}^{ij}(q;\bm a) \psi_{-m-\half}^i \psi_{-n-\half}^{j*} \Big).
\end{equation*}
Now let $r =r_1+r_2+r_3$, and for $1\leq k\leq r$
we denote $i(k)=i$ if $k\in I_i$
and $\bar k = k-\sum_{e=1}^{i(k)-1} r_e$.
Then it is easy to see that this coefficient is:
\begin{equation*}
\begin{split}
 \sum_{\sigma\in S_r} (-1)^\sigma \cdot
\prod_{k=1}^r A_{m_{\bar k}^{i(k)} n_{\overline{\sigma(k)}}^{i(\sigma(k))}}^{i(k)i(\sigma(k))} (q;\bm a)
=\det \big(\tA_{kl} (q; \bm a) \big)_{1\leq k,l\leq r},
\end{split}
\end{equation*}
where $(\tA_{kl})$ is the $r\times r$ matrix given by:
\begin{equation*}
\tA =
\left(
\begin{array}{ccc}
A^{11} & A^{12} & A^{13} \\
A^{21} & A^{22} & A^{23} \\
A^{31} & A^{32} & A^{33} \\
\end{array}
\right),
\end{equation*}
and the block $A^{ij}$ is the matrix of size $r_i\times r_j$
whose entries are $A^{ij}_{m^i_{k} n^j_l} (q;\bm a)$ for $1\leq k\leq r_i$ and $1\leq l\leq r_j$.
Thus we conclude that:
\begin{equation*}
\begin{split}
B_{\mu^1,\mu^2,\mu^3}^{(\bm a)}(q)
=& (-1)^{ \sum_{i=1}^3 \sum_{k=1}^{r_i} n_k^i } \cdot
\det \big(\tA_{kl} (q; \bm a) \big)_{1\leq k,l\leq r} \\
=& \det \big(B_{kl} (q; \bm a) \big)_{1\leq k,l\leq r}.
\end{split}
\end{equation*}
Thus the conclusion holds.
\end{proof}

\subsection{Matching of the two determinantal formulas}
\label{sec-match-coeff}

Recall that in \S \ref{sec-ferm-topovert} we have shown that
the framed topological vertex is a determinant:
\be
W_{\mu^1,\mu^2,\mu^3}^{(a_1,a_2,a_3)} (q) =
\det \big( \tF_{kl}(q; \bm a) \big)_{1\leq k,l\leq r_1+r_2+r_3}.
\ee
Now we compare the entries $\tF_{kl}$ with the entries $B_{kl}$
in the determinantal formula \eqref{eq-prop-det-bog} of the Bogoliubov transform
$B_{\mu^1,\mu^2,\mu^3}^{(\bm a)}(q)$.
From the explicit formulas for $F_{mn}^{ij}$ (see Proposition \ref{prop-coeff-F})
and the explicit formulas \eqref{eqn:def Aqa} for $A_{mn}^{ij}(q;\bm a)$,
we easily see:
\begin{equation*}
\begin{split}
& F_{mn}^{ii}(q;\bm a) = (-1)^n \cdot A_{mn}^{ii}(q;\bm a),
\qquad i=1,2,3; \\
& F_{mn}^{i(i+1)}(q;\bm a) = (-1)^n \cdot q^{(a_i-a_{i+1})/8 + 1/48} \cdot A_{mn}^{i(i+1)}(q;\bm a),
\qquad i=1,2; \\
& F_{mn}^{31}(q;\bm a) = (-1)^n \cdot q^{(a_3-a_1)/8 - 1/24} \cdot A_{mn}^{31}(q;\bm a);\\
& F_{mn}^{i(i-1)}(q;\bm a) = (-1)^n \cdot q^{(a_i-a_{i-1})/8 - 1/48} \cdot A_{mn}^{i(i-1)}(q;\bm a),
\qquad i=2,3; \\
& F_{mn}^{13}(q;\bm a) = (-1)^n \cdot q^{(a_1-a_3)/8 + 1/24} \cdot A_{mn}^{13}(q;\bm a).\\
\end{split}
\end{equation*}
Recall that:
\begin{equation*}
\tF_{kl}(q;\bm a) = F_{m_{\bar k}^i n_{\bar l}^j}^{ij}(q;\bm a),
\qquad
B_{kl} (q;\bm a) = (-1)^{n_{\bar l}^j} \cdot
A_{m_{\bar k}^i n_{\bar l}^j}^{ij}(q;\bm a)
\end{equation*}
for $k\in I_i$ and $l\in I_j$,
and thus:
\be
\label{eq-coeff-F-B}
\begin{split}
& \tF_{kl}(q;\bm a) = B_{kl}(q;\bm a),
\quad \text{ if $k,l\in I_1$ or $k,l\in I_2$ or $k,l\in I_3$;} \\
& \tF_{kl}(q;\bm a) =  q^{(a_i-a_{i+1})/8 + 1/48} \cdot B_{kl}(q;\bm a),
\quad \text{if $k\in I_1,l\in I_2$ or $k\in I_2,l\in I_3$;} \\
& \tF_{kl}(q;\bm a) = q^{(a_3-a_1)/8 - 1/24} \cdot B_{kl}(q;\bm a),
\quad \text{if $k\in I_3,l\in I_1$;}\\
& \tF_{kl}(q;\bm a) = q^{(a_i-a_{i-1})/8 - 1/48} \cdot B_{kl}(q;\bm a),
\quad \text{if $k\in I_2,l\in I_1$ or $k\in I_3,l\in I_2$;} \\
& \tF_{kl}(q;\bm a) = q^{(a_1-a_3)/8 + 1/24} \cdot B_{kl}(q;\bm a),
\quad \text{if $k\in I_1,l\in I_3$.}\\
\end{split}
\ee
Now we show that:
\begin{Lemma}
Let $r=r_1+r_2+r_3$ and let $\sigma \in S_r$ be a permutation.
Then:
\be
\prod_{k=1}^r \tF_{k\sigma(k)} = \prod_{k=1}^r B_{k\sigma(k)}.
\ee
\end{Lemma}
\begin{proof}
The permutation $\sigma$ can be decomposed into a product of cycles:
\begin{equation*}
\sigma = (i_1i_2\cdots i_s) (j_1j_2\cdots j_t) \cdots.
\end{equation*}
We will prove that for each cycle $(i_1i_2\cdots i_s)$,
the following relation holds:
\be
\label{eq-toprove-FB}
\tF_{i_1 i_2} \tF_{i_2i_3}\cdots \tF_{i_{s-1}i_s} \tF_{i_si_1}
= B_{i_1 i_2} B_{i_2i_3}\cdots B_{i_{s-1}i_s} B_{i_si_1},
\ee
and then the conclusion follows from this identity.

For every $k\in \bZ$,
denote by $\tilde k \in \{1,2,3\}$ such that $k\equiv \tilde k(\text{mod } 3)$,
and denote:
\begin{equation*}
\begin{split}
f_{ij} = \begin{cases}
1, &\text{ if $(i,j)=(1,1),(2,2),(3,3)$;}\\
q^{(a_i-a_{i+1})/8 + 1/48}, &\text{ if $(i,j)=(1,2),(2,3)$;}\\
q^{(a_3-a_1)/8 - 1/24}, &\text{ if $(i,j)=(3,1)$;}\\
q^{(a_i-a_{i-1})/8 - 1/48}, &\text{ if $(i,j)=(2,1),(3,2)$;}\\
q^{(a_1-a_3)/8 + 1/24}, &\text{ if $(i,j)=(1,3)$.}
\end{cases}
\end{split}
\end{equation*}
Then the following identity is a simple combinatorial observation:
\begin{equation*}
f_{\tilde i_1 \tilde i_2} f_{\tilde i_2\tilde i_3}\cdots
f_{\tilde i_{s-1}\tilde i_s} f_{\tilde i_s\tilde i_1} =1,
\end{equation*}
which implies \eqref{eq-toprove-FB} by \eqref{eq-coeff-F-B}.
\end{proof}

Now from this lemma we see that:
\begin{equation*}
\begin{split}
\det (\tF_{kl}) =
\sum_{\sigma\in S_r} (-1)^\sigma  \prod_{k=1}^r \tF_{k\sigma(k)}
=  \sum_{\sigma\in S_r} (-1)^\sigma \prod_{k=1}^r B_{k\sigma(k)}
= \det (B_{kl}),
\end{split}
\end{equation*}
and this is equivalent to saying that the framed ADKMV Conjecture holds:
\be
W_{\mu^1,\mu^2,\mu^3}^{(a_1,a_2,a_3)} (q) = B_{\mu^1,\mu^2,\mu^3}^{(a_1,a_2,a_3)} (q).
\ee

\section{Conflict of interest and data availability statement}
The author states that there is no conflict of interest, and
no datasets were generated or analysed during the current study.

\vspace{.2in}

{\em Acknowledgements}.
The authors thank Professor Shuai Guo, Professor Xiaobo Liu and Professor Zhengyu Zong for helpful discussions.
CY thanks Professor Xiangyu Zhou for encouragement.
CY is partly supported by the NSFC (No. 12288201, 12401079),
and the CPSF (No. 2023M743717, BX20240407).
JZ is partly supported by the NSFC (No. 12371254, 11890662, 12061131014).

\end{document}